\documentclass[a4paper,twocolumn,11pt,accepted=2026-07-22]{quantumarticle}
\pdfoutput=1

\usepackage[utf8]{inputenc}
\usepackage[english]{babel}
\usepackage[T1]{fontenc}
\usepackage{amsmath,amssymb,amsfonts,amsthm,mathtools}
\usepackage{bm}
\usepackage{dsfont}
\usepackage{mathbbol}
\usepackage{graphicx}
\usepackage[numbers]{natbib}
\usepackage{hyperref}
\hypersetup{hidelinks}

\theoremstyle{definition}

\newtheorem{corollary}{Corollary}
\newtheorem{definition}{Definition}
\newtheorem{conjecture}{Conjecture}
\newtheorem{lemma}{Lemma}
\newtheorem{proposition}{Proposition}
\newtheorem{theorem}{Theorem}

\newtheorem{example}{Example}

\newcommand{\mbf}{\mathbf}
\newcommand{\mbb}{\mathbb}
\newcommand{\mc}{\mathcal}
\newcommand{\msf}{\mathsf}

\newcommand{\tr}{\operatorname{Tr}}

\newcommand{\ket}[1]{|#1\rangle}

\newcommand{\op}[2]{|#1\rangle\!\langle#2|}

\newcommand{\norm}[1]{\left\lVert#1\right\rVert}

\begin{document}

\title{Quantifiers and witnesses for the nonclassicality of measurements and of states}

\author{Yujie Zhang}
\email{yujie4physics@gmail.com}
\affiliation{Institute for Quantum Computing, University of Waterloo, Waterloo, Ontario Canada N2L 3G1}
\affiliation{Department of Physics \& Astronomy, University of Waterloo, Waterloo, Ontario, Canada, N2L 3G1}
\affiliation{Perimeter Institute for Theoretical Physics, 31 Caroline Street North, Waterloo, Ontario Canada N2L 2Y5}
\author{Y{\`i}l{\`e} Y{\=\i}ng}
\email{yile.ying@gmail.com}
\affiliation{Perimeter Institute for Theoretical Physics, 31 Caroline Street North, Waterloo, Ontario Canada N2L 2Y5}
\affiliation{Department of Physics \& Astronomy, University of Waterloo, Waterloo, Ontario, Canada, N2L 3G1}
\author{David Schmid} 
\email{davidschmid10@gmail.com}
\affiliation{Perimeter Institute for Theoretical Physics, 31 Caroline Street North, Waterloo, Ontario Canada N2L 2Y5}

\maketitle

\begin{abstract}
 
In recent work~\cite{zhang2024parellel}, we proposed a unified notion of nonclassicality that applies to arbitrary processes in quantum theory, including individual quantum states, measurements, and sets thereof. This notion is derived from the principle of generalized noncontextuality, but in a novel manner that applies to individual processes rather than full experiments or theories.
In the present work, we develop semidefinite-programming-based certificates and witnesses for the nonclassicality of states, sources, measurements, and sets thereof. These theory-dependent methods complement theory-independent approaches based on noncontextuality inequalities. We demonstrate the framework through a variety of explicit examples.
\end{abstract}
\tableofcontents 

\section{Introduction}

A principled way to demonstrate that a theory or experiment resists classical explanation is to prove that it cannot be represented in any generalized-noncontextual ontological model~\cite{gencontext,Schmid2024structuretheorem}. This approach can be motivated by a methodological version of Leibniz's principle of the identity of indiscernibles~\cite{Leibniz}. It is also equivalent to the impossibility of a positive quasiprobability representation~\cite{negativity,SchmidGPT,Schmid2024structuretheorem}. Furthermore, it coincides with the natural notion of classical explainability in GPTs~\cite{SchmidGPT,Schmid2024structuretheorem} arising in the framework of generalized probabilistic theories (GPTs)~\cite{Hardy,GPT_Barrett}.

In most works, the principle of generalized noncontextuality is applied to {\em experimental phenomena}, as a rigorous criterion for determining whether or not they are classically explainable. For any observed phenomenon, one can determine whether or not it admits of a noncontextual explanation. Aided by a growing toolkit of analytical methods~\cite{Schmid2018,Chaturvedi2021characterising,schmid2024PTM,Schmid2024structuretheorem,SchmidGPT,catani2024resource,PuseydelRio,schmid2024shadows,schmid2024addressing,mazurek2016experimental,grabowecky2021experimentally}, such investigations have been done in the areas of
quantum computation~\cite{Schmid2022Stabilizer,shahandeh2021quantum}, state discrimination~\cite{schmid2018contextual,flatt2021contextual,mukherjee2021discriminating,Shin2021}, interference~\cite{Catani2023whyinterference,catani2022reply,catani2023aspects,giordani2023experimental}, compatibility~\cite{Selby2023,selby2023accessible,Tavakoli2020}, uncertainty relations~\cite{catani2022nonclassical}, metrology~\cite{contextmetrology}, thermodynamics~\cite{contextmetrology,comar2024contextuality,lostaglio2018}, weak values~\cite{AWV, KLP19}, coherence~\cite{rossi2023contextuality,Wagner2024coherence,wagner2024inequalities}, quantum Darwinism~\cite{baldijao2021noncontextuality}, information processing and communication~\cite{POM,RAC,RAC2,Saha_2019,Yadavalli2020,PhysRevLett.119.220402,fonseca2024robustness}, cloning~\cite{cloningcontext}, broadcasting~\cite{jokinen2024nobroadcasting}, pre- and post-selection paradoxes~\cite{PP1}, randomness certification~\cite{Roch2021}, psi-epistemicity~\cite{Leifer}, and Bell~\cite{Wright2023invertible} and Kochen-Specker scenarios~\cite{operationalks,kunjwal2018from,Kunjwal16,Kunjwal19,Kunjwal20,specker,Gonda2018almostquantum}. 
Generalized noncontextuality also gives a rigorous criterion for determining whether or not {\em a full theory} is classically explainable or not. Some work has been done to characterize what can be said about generalized noncontextuality at the level of full theories~\cite{Schmid2024structuretheorem,SchmidGPT,Schmid2024reviewreformulation,schmid2020unscrambling}.

In Ref.~\cite{zhang2024parellel}, we showed how noncontextuality can also be used to induce a notion of nonclassicality {\em at the level of an individual quantum process}---for example, for a single state, measurement, or channel. The basic idea is simple: a quantum process is nonclassical if and only if it can be leveraged in a nontrivial way within {\em some} quantum circuit to generate data that cannot be reproduced in any generalized-noncontextual ontological model. 

The resulting boundary between classical and nonclassical does not always coincide with what one might intuitively expect based on traditional notions of classicality found in the literature. For instance, we showed in Ref.~\cite{zhang2024parellel} that, while it is true that all entangled states, incompatible sets of measurements, and entanglement-non-breaking channels are nonclassical according to our proposal, some separable states, compatible sets of measurements, and entanglement-breaking channels are {\em also} nonclassical. This raises the question of characterizing the classical-nonclassical boundary for individual quantum processes. First results regarding this boundary are given in Ref.~\cite{zhang2024parellel}.

One would moreover like to have {\em quantitative} methods for certifying and quantifying the nonclassicality of any given process. Here, we provide several semidefinite programs (SDPs) for this purpose. By considering the duals to these SDPs, we construct witnesses for the nonclassicality of a measurement or of a state. However, these SDP-based witnesses rely on the validity of quantum theory, analogous to how entanglement witnesses are not device-independent~\cite{Brunner2014}. This complements the derivation of theory-independent means of certifying the nonclassicality of a given measurement or state using noncontextuality inequalities~\cite{Schmid2018, mazurek2016experimental,schmid2024PTM}, analogous to how Bell inequalities are device-independent. Using known tools for deriving such inequalities, we show how one can certify entanglement of bipartite states that do not violate any steering (or Bell) inequalities. 

We apply these methods to many specific examples, and in particular to study the nonclassicality of some quantum measurements and collections of states that traditionally would have been considered to be classical.

\section{Preliminaries}
\label{sec:II}

We give a brief introduction to generalized noncontextuality in the simplest case of prepare-and-measure scenarios~\cite{gencontext}, adapted from the formulation in our companion work~\cite{zhang2024parellel}. As in that work, here we will also focus only on quantum theory, even though our approach can be immediately generalized to any generalized probabilistic theory~\cite{Hardy,GPT_Barrett}. As shown in Ref.~\cite{SchmidGPT}, there exists a noncontextual model for a prepare-and-measure scenario if and only if there exists a linear and diagram-preserving ontological model for the GPT representation of that scenario. 
The GPT representation of quantum theory is just the familiar one from quantum information theory, where a \textit{preparation} is associated with a quantum state $\rho$ and a measurement is associated with a positive operator-valued measure (POVM) $\msf{M}\coloneqq\{M_b\}_b$. 

In this work, we consider also preparation processes such as sources (probabilistic ensembles of states), denoted $\msf{P}\coloneqq\{p(a)\rho_a\}_a$, multi-states (sets of states), denoted $\msf{P}\coloneqq\{\rho_x\}_x$, and most generally multi-sources (sets of ensembles of states), denoted $\msf{P}\coloneqq\{\{p(a|x)\rho_{a|x}\}_{a}\}_{x}$; similarly, we consider measurements, denoted $\msf{M}\coloneqq\{M_{b}\}_{b}$, and sets of measurements, which we call multi-measurements, denoted $\msf{M}\coloneqq\{\{M_{b|y}\}_{b}\}_{y}$. These types of processes are pictured and discussed further in Ref.~\cite{zhang2024parellel}. 

\begin{figure}[htb!]
\centering
\includegraphics[width=0.4\textwidth]{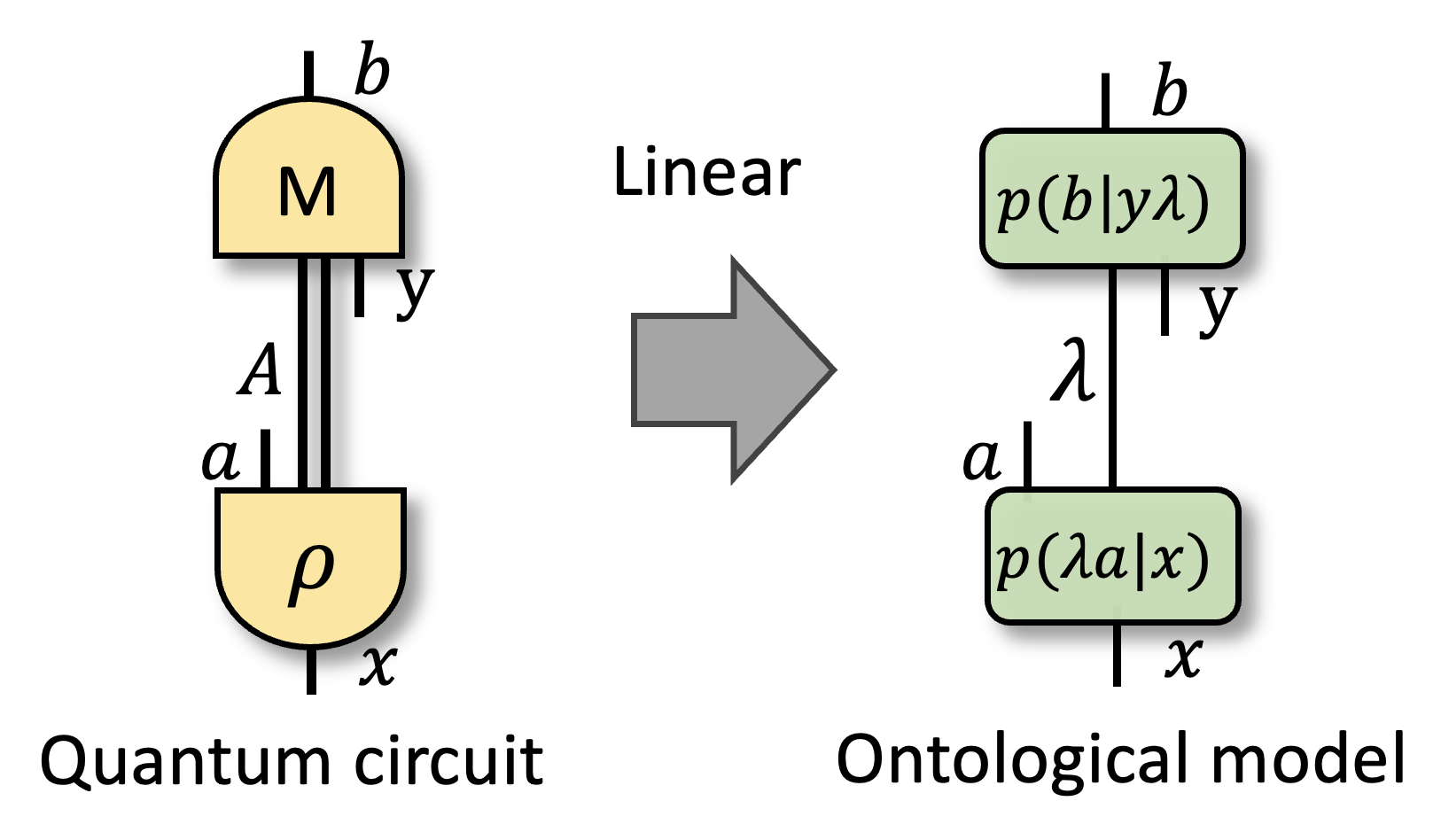}
\caption{ A prepare-and-measure circuit (left) involving the composition of a multi-source and a multi-measurement, and an ontological model for it (right).}
\label{PMfig}
\end{figure}

The most general prepare-and-measure scenario in quantum theory involves implementing a multi-measurement on the output of a multi-source, as shown in Figure~\ref{PMfig}. The quantum predictions for such a scenario are given by
\begin{equation}
p(ab|xy)=p(a|x)\tr[ M_{b|y}\rho_{a|x}].
\label{eq:pmstats}
\end{equation}

 In such a scenario, diagram-preservation implies that each normalized state $\rho_{a|x}$ is represented by a probability distribution $p(\lambda|a,x)$ over a set $\Lambda$ of ontic states, while each effect $M_{b|y}$ is represented by a response function $p(b|y,\lambda)$. Equivalently, the subnormalized state $p(a|x)\rho_{a|x}$ is represented by the joint distribution
\begin{equation}
p(a,\lambda|x)\coloneqq p(a|x)p(\lambda|a,x).
\label{eq:jointprep}
\end{equation}
These ontological quantities must reproduce the quantum statistics according to
\begin{equation}
p(ab|xy)=\sum_{\lambda} p(b|y,\lambda)p(a,\lambda|x).
\label{eq:PMNC}
\end{equation}

The states and effects in such a circuit can satisfy certain nontrivial linear relations, which we call \emph{operational identities}. We identify these operational identities by their coefficients, and denote the set of all operational identities for the preparation procedures and that for the measurement procedures, respectively, by
\begin{subequations}
\begin{align}
 \mc{O}(\msf{P})
 &\coloneqq
 \Bigl\{\{\alpha_{a,x}\}_{a,x}\,\Big|\, \sum_{a,x}\alpha_{a,x}\,p(a|x)\rho_{a|x}=0\Bigr\},
 \label{eq:op-prep} \\
 \mc{O}(\msf{M})
 &\coloneqq
 \Bigl\{\{\beta_{b,y}\}_{b,y}\,\Big|\, \sum_{b,y}\beta_{b,y}\,M_{b|y}=0\Bigr\}.
 \label{eq:op-meas}
\end{align}
\end{subequations}
If the ontological representation is linear, then these operational identities must also be respected at the ontological level. That is, for every ontic state $\lambda$, we must have
\begin{subequations}
\begin{align}
\sum_{a,x}\alpha_{a,x}\,p(a,\lambda|x)&=0
\qquad \forall\,\{\alpha_{a,x}\}_{a,x}\in \mc{O}(\msf{P}),
\label{eq:constrain-prep}\\
\sum_{b,y}\beta_{b,y}\,p(b|y,\lambda)&=0
\qquad \forall\,\{\beta_{b,y}\}_{b,y}\in \mc{O}(\msf{M}).
\label{eq:constrain-meas}
\end{align}
\end{subequations}

\begin{definition}
\label{def:stat}
A prepare-and-measure experiment with a quantum multi-source $\msf{P}=\{\{p(a|x)\rho_{a|x}\}_{a}\}_{x}$ and a quantum multi-measurement $\msf{M}=\{\{M_{b|y}\}_{b}\}_{y}$ is classically explainable if and only if there exists an ontological model with ontic state space $\Lambda$, preparation assignments $p(a,\lambda|x)$, and response functions $p(b|y,\lambda)$ such that Eqs.~\eqref{eq:PMNC}, \eqref{eq:constrain-prep}, and~\eqref{eq:constrain-meas} are satisfied. 
\end{definition}

 In other words, a prepare-and-measure experiment is classically explainable precisely when its statistics can be reproduced by a linear, diagram-preserving ontological representation~\cite{Schmid2024structuretheorem,zhang2024parellel}. 
\section{Nonclassicality of a multi-measurement and a measurement}
\label{sec:III}
Let us begin by defining the nonclassicality of a quantum measurement, or of a set of quantum measurements, following Ref.~\cite{zhang2024parellel}. 
\begin{definition}
\label{def: nonclassical measurement}
A multi-measurement is classical if and only if the statistics generated by the set of circuits where it is contracted with {\em any} state (i.e., where the set ranges over all states) are classically explainable in the sense of Def.~\ref{def:stat}.
\end{definition}

Let us stress the conservative, compositional character of this definition. The sense in which a measurement process is classical here is that it is stable under arbitrary composition with processes of the dual type: no
choice of quantum state can make it participate in operational statistics that fail to be classically explainable. Thus, if a multi-source $\msf P$ and a multi-measurement $\msf M$ generate prepare-and-measure statistics that are not classically explainable, then $\msf M$ is necessarily nonclassical in the present sense. If $\msf M$ were classical, its composition with every quantum state, and hence with the states arising from any source, would be classically explainable. The same nonclassical statistics necessarily also witness
the nonclassicality of $\msf P$ according to the dual definition below. This is not an artificial feature of the definition, but rather the intended consequence of defining process classicality through stability under
arbitrary composition with dual processes.

To begin characterizing nonclassicality of measurements, consider first the noncontextual measurement-assignment polytope for a given set $\{\{M_{b|y}\}_b\}_y$ of measurements satisfying operational identities $\mc{O}({\msf{M}})$:
\begin{definition}
\label{def: op_polytope}
The noncontextual measurement-assignment polytope $\mbb{P}_{\msf{M}}$ associated with a multi-measurement $\msf{M}=\{\{M_{b|y}\}_b\}_y$ is the set of points $\{p(b|y)\}_{b,y}$ that satisfy the following constraints:
\begin{subequations}
\begin{align}
\text{(i)}~~&p(b|y) \geq 0~~~~\forall b,y \label{eq:cons_pos}\\
\text{(ii)}~~&\sum_b p(b|y)=1~~~~\forall y \label{eq:cons_norm}\\
\text{(iii)}~~&\sum_{b,y} \beta_{b,y} p(b|y)=0 \quad\quad \forall \{\beta_{b,y}\}\in \mc{O}({\msf{M}})
\end{align}
\end{subequations}
where $\mc{O}({\msf{M}})$ is defined by the operational identities holding among effects, as in Eq.~\eqref{eq:op-meas}. 
\end{definition}

 Then, following Theorem~3 of Ref.~\cite{zhang2024parellel}, the characterization of classical sets of measurements can be equivalently rewritten in terms of the noncontextual measurement-assignment polytope $\mathbb{P}_{\mathsf M}$ defined in Def.~\ref{def: op_polytope} as follows. 
\begin{theorem}\label{theoremmeas}
A set of measurements $\msf{M}=\{\{M_{b|y}\}_b\}_y$ is classical if and only if
its effects can be decomposed as
 \begin{align}
 &M_{b|y}=\sum_{\lambda} p(b|y,\lambda)G_{\lambda} ~~~~\forall b,y \label{eq_Edecom1} \\ 
 \text{with}~~&\{p(b|y,\lambda)\}_{b,y}\in \mbb{P}_{\msf{M}}~~~~\forall \lambda \label{eq_Edecom2}
 \end{align}
 where $\{G_{\lambda}\}_{\lambda}$ is a POVM and $\mbb{P}_{\msf{M}}$ is the polytope defined in Def.~\ref{def: op_polytope}.
\end{theorem}
 
\begin{proof}
 For completeness, we include a brief proof, adapted from the proof of Theorem~3 in Ref.~\cite{zhang2024parellel}.

If the multi-measurement $\msf{M}$ is classical, then by Def.~\ref{def: nonclassical measurement}, for every quantum state $\rho$,
\begin{align}
\tr[M_{b|y}\rho]=\sum_{\lambda} p(b|y,\lambda)p(\lambda|\rho)
\qquad \forall b,y,\rho,
\end{align}
where, for each $\lambda$, the conditional probabilities $\{p(b|y,\lambda)\}_{b,y}$ satisfy the ontological identities associated with $\mc{O}(\msf{M})$. Since $p(\lambda|\rho)$ is linear in $\rho$ for all $\rho$, by the Riesz representation theorem, there exists a positive semidefinite operator $G_\lambda$ such that
\begin{align}
p(\lambda|\rho)=\tr[G_\lambda\rho]
\qquad \forall \rho.
\end{align}
Since $\sum_\lambda p(\lambda|\rho)=1$ for all $\rho$, it follows that $\sum_\lambda G_\lambda=\mbb{1}$, so $\{G_\lambda\}_\lambda$ is a POVM. Therefore,
\begin{align}
\tr[M_{b|y}\rho]
=\tr\!\left[\Bigl(\sum_\lambda p(b|y,\lambda)G_\lambda\Bigr)\rho\right]
\qquad \forall b,y,\rho,
\end{align}
and hence
\begin{align}
M_{b|y}=\sum_\lambda p(b|y,\lambda)G_\lambda
\qquad \forall b,y.
\end{align}
Since each $\{p(b|y,\lambda)\}_{b,y}$ is nonnegative, normalized, and respects $\mc{O}(\msf{M})$, we have $\{p(b|y,\lambda)\}_{b,y}\in\mbb{P}_{\msf{M}}$ for all $\lambda$.

Conversely, if Eqs.~\eqref{eq_Edecom1}--\eqref{eq_Edecom2} hold, define
\begin{align}
p(\lambda|\rho)\coloneqq \tr[G_\lambda\rho].
\end{align}
Then $p(\lambda|\rho)$ is a valid probability distribution, linear in $\rho$, and
\begin{align}
\tr[M_{b|y}\rho]
=\sum_\lambda p(b|y,\lambda)p(\lambda|\rho).
\end{align}
Because $\{p(b|y,\lambda)\}_{b,y}\in\mbb{P}_{\msf{M}}$, the ontological identities associated with $\mc{O}(\msf{M})$ are satisfied. Hence, the statistics generated by contracting $\msf{M}$ with arbitrary states are classically explainable, so $\msf{M}$ is classical.
\end{proof}

One can equivalently 
rewrite this condition in terms of the extreme points 
$D_{\mbb{P}_{\msf{M}}}\in \mbb{P}_{\msf{M}}$ of the polytope in Def.~\ref{def: op_polytope}. (These extreme points can be explicitly computed by solving the vertex
enumeration problem~\cite{Christof1997, Avis1997}.)
In particular, a set of measurements $\msf{M}=\{\{M_{b|y}\}_b\}_y$ is classical if and only if its effects can be decomposed as
\begin{align}\label{usefulrechar}
M_{b|y}=\sum_{\lambda} D_{\mbb{P}_{\msf{M}}}(b|y,\lambda)G_{\lambda},
\end{align}
where $\{G_{\lambda}\}_{\lambda}$ is a POVM. 
This is because we can decompose each individual response function in terms of extremal points of the polytope as $p(b|y,\lambda)=\sum_{\lambda'}D_{\mbb{P}_{\msf{M}}}(b|y,\lambda')p(\lambda'|\lambda)$, so that
\begin{align} \label{eqdetm}
M_{b|y}&=\sum_{\lambda} p(b|y,\lambda)G_{\lambda}\\
&=\sum_{\lambda} \sum_{\lambda'}D_{\mbb{P}_{\msf{M}}}(b|y,\lambda')p(\lambda'|\lambda)G_{\lambda}\\
&=\sum_{\lambda'}D_{\mbb{P}_{\msf{M}}}(b|y,\lambda')G'_{\lambda'},
\end{align}
where we defined $G'_{\lambda'}=\sum_{\lambda}p(\lambda'|\lambda)G_{\lambda} $ (which also constitutes a POVM).

As noted in corollary~5 of Ref.~\cite{zhang2024parellel}, Theorem~\ref{theoremmeas} immediately implies that every set of incompatible measurements is nonclassical, since Eq.~\eqref{eq_Edecom1} describes the usual condition for a set of measurements to be compatible---namely, that there exists a single POVM $\{G_{\lambda}\}_{\lambda}$ which can be postprocessed to recover any of those measurements. But in Theorem~\ref{theoremmeas}, one does not allow arbitrary postprocessings, but only those consistent with Eq.~\eqref{eq_Edecom2}.
So, unlike incompatibility, this notion of nonclassicality is nontrivial even for a single measurement $\{M_b\}_b$.

\begin{corollary}
\label{coro: nc-meas}
A single measurement $\msf{M}=\{M_{b}\}_b$ is classical if and only if its effects can be decomposed as
 \begin{align}
 \label{eq_Edecom}
 &M_{b}=\sum_{\lambda} D_{\mbb{P}_{\msf{M}}}(b|\lambda)G_{\lambda},
 \end{align}
 where $\{G_{\lambda}\}_{\lambda}$ is a POVM and $D_{\mbb{P}_{\msf{M}}}(b|\lambda)$ is an extreme point of the polytope defined in Def.~\ref{def: op_polytope} (for the case when $y$ is trivial).
\end{corollary}

In fact, the qualitative nonclassicality of any set of measurements can be deduced by studying a single related measurement, using the idea of flag-convexification introduced in Ref.~\cite{Selby2023}.

In flag-convexification, the classical input variable encoding one's choice of measurement has its value sampled according to some full-support probability distribution, and then this value is copied and encoded in a new classical output variable. We depict an example of this on the right-hand side of Figure~\ref{flag_conv}. (In this example and henceforth, we take the probability distribution in question to be the uniform distribution, although our results hold just as well for any other full-support distribution.)
\begin{figure}[htb!]
\centering
\includegraphics[width=0.45\textwidth]{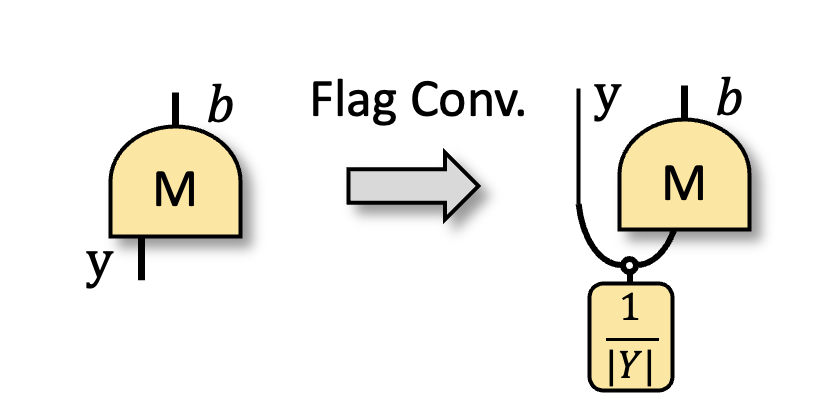}
\caption{A multi-measurement (left) and its flag-convexification with a uniform distribution $\frac{1}{|Y|}$ (right).}
\label{flag_conv}
\end{figure}

\begin{proposition}
A set of measurements $\{\{M_{b|y}\}_{b}\}_{y}$ is nonclassical if and only if its flag-convexification $\{\tilde{M}_{b,y}\coloneqq \frac{1}{|Y|}M_{b|y}\}_{b,y}$ is nonclassical. \label{prop:to_single_model}
\end{proposition}

 This was implicitly established in Ref.~\cite{Selby2023}, and the analogous result for non-uniform flag-convexification was proven in Ref.~\cite{selby2023accessible}. Specifically, Ref.~\cite{Selby2023} showed that, in a prepare-and-measure scenario, the statistics generated by a measurement set $\{\{M_{b|y}\}_b\}_y$ together with an arbitrary preparation set $\mathsf P$ are classically explainable if and only if the statistics generated by the flag-convexified measurement $\{\tilde M_{b,y}:=\frac{1}{|Y|}M_{b|y}\}_{b,y}$ together with $\mathsf P$ are classically explainable. Hence, by Def.~\ref{def: nonclassical measurement}, qualitative nonclassicality is preserved under flag-convexification.

Proposition~\ref{prop:to_single_model} does not say anything about the {\em quantitative} amount of nonclassicality in a multi-measurement under flag-convexification. We show in Appendix~\ref{appfc1} and Appendix~\ref{appfc2} that many of the measures we introduce in this work do not change under flag-convexification.
We do not expect such a quantitative equivalence to hold for general measures of nonclassicality.

\section{Certifying and quantifying the nonclassicality of a measurement}
\label{sec:IV}
In this section, we will focus solely on certifying and quantifying the nonclassicality of a single measurement. These methods can {\em also} be applied directly to quantify arbitrary sets of measurements (i.e., multi-measurements), since (as we will show in Appendix~\ref{appfc1} and Appendix~\ref{appfc2}) the quantifiers we introduce here give the same values for a multi-measurement as for the single measurement generated by that multi-measurement under flag-convexification. 

 The methods introduced in the following first two subsections are theory-dependent: they certify nonclassicality relative to the quantum description of the processes involved, rather than in a device-independent manner. Theory-independent certification will be discussed separately in Sec.~\ref{sec:IVC}.

Similar to the characterization methods used extensively to study quantum steering and standard measurement incompatibility~\cite{Cavalcanti2017}, the nonclassicality of a given measurement $\{M_b\}_b$ can be certified numerically by solving a semidefinite program. This program follows immediately from Corollary~\ref{coro: nc-meas}: 
\begin{align}
\label{sdp1}
\max_{\{G_{\lambda}\}_{\lambda}}&~\mu \notag \\
\text{s.t. }&\sum_{\lambda}D_{\mbb{P}}(b|\lambda)G_{\lambda}=M_b\quad\quad \forall b \notag \\
&G_{\lambda}\ge \mu\mbb{1} \quad \forall \lambda.
\end{align}
 (Note that the condition that $\sum_\lambda G_\lambda = \mbb{1}$ follows from the first constraint in the SDP, as one can see by summing both sides over $b$.)
If the program returns an optimal value $\mu$ that is negative, then it follows that there is no decomposition as in Corollary \ref{coro: nc-meas} where $G_\lambda \geq 0$, and so $\{M_b\}_b$ is nonclassical.
If, however, the given measurement $\{M_b\}_b$ is classical, the solution to the SDP gives the parent POVM $\{G_{\lambda}\}_{\lambda}$ and response functions for its simulation. 

\subsection{Robustness-based quantifier for nonclassical measurement}

One can also {\em quantify} the nonclassicality of a given measurement, using techniques like those used for entanglement \cite{Guhne2009}, quantum steering \cite{Uola2020review}, and many other quantum resources \cite{Chitambar2019}. In what follows, we will discuss both robustness-based and weight-based quantifiers. We leave the more ambitious task of developing a full resource-theoretic framework for understanding nonclassicality for future work.
The programs we introduce here also yield optimal linear witnesses for the certification of resources.

Quantifiers of nonclassicality based on robustness ask how much noise must be added to a given measurement for it to become classical. Depending on the noise model, one can define different measures such as the generalized robustness, standard robustness, and random robustness~\cite{Cavalcanti2017,Chitambar2019}. Here, we consider the white-noise robustness (which is also termed as resource random robustness~\cite{Chitambar2019}), which asks how much white noise would need to be added for the nonclassicality of a given measurement to be completely destroyed. Given a measurement $\{M_b\}_b$, one defines a noisy POVM with the effects
\begin{equation}
M_{b}^{\eta} \coloneqq \eta M_{b} + (1 - \eta) \frac{\tr[M_{b}]}{d}\mathbb{1},
\end{equation}
where $d$ is the dimension of the Hilbert space; the critical parameter $\eta$ at which the transition to classicality occurs is the white-noise robustness. 
Notice that---following an awkward but standard convention---lower values correspond to higher resourcefulness. All classical measurements have the maximum value $1$. 

The white-noise robustness can be computed (following, e.g., Refs.~\cite{Heinosaari2015,Designolle2019}) using the following semidefinite program: 
\begin{align}
\eta_{\{M_b\}}= \max_{\{G_{\lambda}\}_{\lambda}} &\eta \notag \\
\text{s.t. }&\sum_{\lambda}D_{\mbb{P}}(b|\lambda)G_{\lambda}=M^{\eta}_b\quad\quad \forall b \notag \\
& \eta\le 1,~G_{\lambda}\ge \mbf{0}\quad \forall \lambda \quad
\label{pri:SDP}
\end{align}
where the $D_{\mbb{P}}(b|\lambda)$ are the extreme points in the polytope $\mbb{P}$ of Def.~\ref{def: op_polytope} (for the case when $y$ is trivial).

This SDP can be efficiently computed. One can also get an analytical upper bound on this quantifier by studying the dual of the primal problem above, which is
\begin{align}
\eta_{\{M_b\}}
&= \min_{\{X_b\}_b}\left(1+\sum_b\tr[X_bM_b]\right) \notag\\
\text{s.t.}\quad
&1+\sum_b\tr[X_bM_b] \notag\\
&\qquad\geq \frac{1}{d}\sum_b \tr X_b\,\tr M_b, \notag\\
&\sum_b D_{\mbb P}(b|\lambda)X_b\geq 0
\quad \forall\lambda.
\end{align}
By a standard result in semidefinite programming, the solution to this dual is greater than or equal to the solution $\eta$ of the primal program. It follows that any specific feasible solution $\{X_{b}\}_b$ of the above dual problem provides an upper bound on $\eta_{\{M_b\}}$. By optimizing a cleverly chosen family of dual variables of the form $X_b=\alpha\mbb{1}-\beta M_b$ (where $\alpha$ and $\beta$ are constrained so that $X_b$ is a feasible solution), one can obtain an analytical upper bound on the white-noise robustness~\cite{Designolle2019}, namely 
\begin{align}
 \eta_{\{M_b\}}\le \frac{d^2\Lambda-\sum_b (\tr{M_b})^2}{\sum_b [d{\tr M_b^2}- ({\tr M_b})^2]},
\end{align}
where $\Lambda=\max_{\lambda}\norm{\sum_b D_{\mbb{P}}(b|\lambda)M_b}_{2}$ and where $\norm{X}_{2}$ is the spectral norm defined as the largest absolute value of the eigenvalues of $X$. For a rank-1 POVM $\{M_{b}\}_{b=1}^{k}$ with $k$ effects having equal trace $\tr M_b=\frac{d}{k}$, this further simplifies to 
\begin{align}
\eta\le \frac{k\Lambda - 1}{d-1}.
\label{eq:etaupper}
\end{align}
We expand on this argument in Appendix~\ref{appendix: upperbound}.

We now give some examples and compute their white-noise robustness.
Surprisingly, we found the upper bound just given in Eq.~\eqref{eq:etaupper} to be tight for all of these measurements, a fact
we showed numerically by lower-bounding it with the primal SDP in Eq.~\eqref{pri:SDP}.

\begin{example}
\label{bb84}
Consider a measurement composed of $k$ effects arranged symmetrically in a plane, namely the measurement $\msf{M}_k\coloneqq\{M_b\}_{b=1}^k$ with effects 
\begin{equation}
 M_b=\frac{1}{k}[\mbb{1}+\cos\theta_b\sigma_x+\sin\theta_b\sigma_z], \quad \theta_b=\frac{2\pi b}{k}.
 \label{eq: planar}
\end{equation}
The white-noise robustness $\eta_k$ of this measurement 
is
\begin{align}
 \eta_{k}= \frac{1}{2\cos(\pi/k)} \quad\quad \forall k\ge 4. 
\end{align}
(When $k<4$, such a measurement is always classical.) This value can be found using the upper bound in Eq.~\eqref{eq:etaupper}, and then checked to be tight by constructing an explicit simulation model in Eq.~\eqref{pri:SDP}. For small $k$, one can check explicitly that the optimal $\{G_{\lambda}\}_\lambda$ for simulating the noisy $\{M_b\}_{b=1}^k$ is given by the measurement itself when $k$ is odd 
\begin{align}
G_{\lambda}=\frac{1}{k}[\mbb{1}+\cos\theta_{\lambda}\sigma_x+\sin\theta_{\lambda}\sigma_z], \quad \theta_{\lambda}=\frac{2\pi {\lambda}}{k}
\end{align}
 and by
\begin{align}
G_{\lambda}
&=\frac{1}{k}[\mbb{1}+\cos\theta_{\lambda}\sigma_x
+\sin\theta_{\lambda}\sigma_z],\notag\\
\theta_{\lambda}&=\frac{2\pi\lambda+\pi}{k}.
\end{align}
when $k$ is even.\par 
Both the pentagon measurement in~\cite{Selby2023} and the BB84 measurement are special cases of such a $k$-outcome symmetric planar measurement.
\label{example1}
\end{example}

\begin{table}[htb!]
\centering
\caption{White-noise robustness $\eta_k$ for various $k$-outcome symmetric planar qubit measurements, as defined in Eq.~\eqref{eq: planar}.}
\label{tab:eta_values}
\begin{tabular}{rr}
\hline
\hline
$k$ &\quad $\eta_k$ \\ 
\hline
3 & \quad 1\\ 
4 &\quad $\frac{\sqrt{2}}{2} \approx 0.707$ \\
5 & \quad$\frac{\sqrt{5}-1}{2} \approx 0.618$ \\
6 &\quad $\frac{\sqrt{3}}{3} \approx 0.577$ \\
7 & \quad$\frac{1}{2\cos(\frac{\pi}{7})} \approx 0.555$ \\
8 & \quad$\sqrt{\frac{2-\sqrt{2}}{2}} \approx 0.541$ \\
\hline
\hline
\end{tabular}
\end{table}

\begin{example}
Consider a measurement composed of $k$ effects that are the vertices of a platonic solid embedded in the qubit (centered, and with effects rescaled appropriately)~\cite{zhang2025cost}. The white-noise robustness $\eta$ of such measurements was computed in the same manner as for Example~\ref{bb84}, and is shown in Table~\ref{tab:eta_Platonic}. The optimal $\{G_{\lambda}\}_{\lambda}$ for simulating different noisy platonic measurements is given by the polytope dual\footnote{ \label{dual} The vertices of the dual polytope are given by the set of unit vectors from the origin that are normal to each face of the given polytope.}; e.g., the optimal $\{G_{\lambda}\}_{\lambda}$ for simulating a noisy cubic POVM is the octahedron POVM. 
\label{example2}
\end{example}

\begin{table}[h]
\centering
\caption{White-noise robustness $\eta_k$ for various $k$-outcome Platonic solid measurements, as defined in Example~\ref{example2}. Codes are available at \cite{zhang2024gita}.}
\label{tab:eta_Platonic}
\begin{tabular}{cr}
\hline
\hline
 $\#$ of Vertices & $\eta_v^{\text{Plat}}\quad$ \\ 
\hline
 4& \quad $1$ \quad\quad \\
\hline
 6& \quad $\frac{\sqrt{3}}{3}\approx 0.577$ \\
\hline
8 &\quad $\frac{\sqrt{3}}{3}\approx 0.577$ \\
\hline
 12 &\quad $ \sqrt{\frac{5-2\sqrt{5}}{3}}\approx 0.4195 $ \\
\hline
 20 &\quad $\sqrt{\frac{5-2\sqrt{5}}{3}} \approx 0.4195$\\
 \hline
\hline
\end{tabular}
\end{table}

\begin{example}
Consider a set of measurements in mutually unbiased bases, for dimensions $d \in \{2,3,4\}$. The white-noise robustness $\eta$ of each such set of measurements was computed in the same manner as for Example~\ref{bb84}, and examples are shown in Table~\ref{tab:eta_MUB}. Note that, as a consequence of Proposition~\ref{prop:to_single_model}, one would in each case obtain the same white-noise robustness for the single measurement constructed by flag-convexifying these multi-measurements.
\end{example}
\begin{table}[htb!]
\centering
\caption{ White-noise robustness for the complete set of $d+1$ MUB measurements in dimension $d$.}
\label{tab:eta_MUB}
\begin{tabular}{cr}
\hline
\hline
Dim & $\eta_{d}$\quad \\ 
\hline
$2$ & $\frac{\sqrt{3}}{3}\approx0.577$ \\
\hline
$3$ & $\frac{1+3\sqrt{5}}{16}\approx 0.4818$ \\
\hline
$4$ & $\frac{3+2\sqrt{3}}{15}\approx 0.4309 $ \\

\hline
\hline
\end{tabular}
\end{table}

\subsection{The nonclassical fraction of a measurement}\label{ncfraction}

A second class of measures commonly used to quantify quantum resources consists of \textit{weight-based quantifiers}~\cite{Skrzypczyk2014,Lewenstein1998,Chitambar2019}.
Such measures are defined as the minimal cost of generating the measurement in question by mixing together an arbitrary measurement with an arbitrary classical measurement, such that the measurements being mixed satisfy the same operational identities as the given measurement. That is, the effects in any given measurement $\{M_b\}_b$ can be written as
\begin{equation}
 M_b=\omega N_b+ (1-\omega)K_b, 
\end{equation}
where $0\le \omega\le 1$,
where $\{N_b\}_b$ is an arbitrary measurement while $\{K_b\}_b$ is classical, and where both $\{N_b\}_b$ and $\{K_b\}_b$ satisfy the same operational identities as $\{M_b\}_b$.
The minimal value of $\omega$ that can be used in such a decomposition is a weight-based quantifier of nonclassicality of a measurement, and we term it the {\em nonclassical fraction} (analogous to, e.g., the nonlocal fraction~\cite{Wolfe2020quantifyingbell}). Higher values correspond to higher nonclassicality, with a maximum of $1$ (and all classical resources have value $0$). 

Note that we do not allow $\{N_b\}_b$ and $\{K_b\}_b$ to be an arbitrary measurement and an arbitrary classical measurement, respectively. 
This is because the nonclassical set of measurements is nonconvex, and very little is known about weight-based quantifiers in nonconvex contexts~\cite{Kuroiwa2024,Kuroiwa2024b}. By stipulating that the sets of measurements to be mixed satisfy the same operational equivalences as the measurement to be decomposed, we return to a convex setting.
As such, it is better viewed as a collection of measures, rather than a single measure on the space of all multi-measurements. In any case, our main purpose for introducing this quantity is for the construction of nonclassicality witnesses, as we will discuss.
Whether or not the quantity defined {\em without} this stipulation is useful and interesting remains an open question.

The nonclassical fraction can be computed using the following SDP:
\begin{align}
\omega=\min_{\{{G}_{\lambda}\}_{\lambda}} &1-\frac{\tr\sum_{\lambda}{{G}_{\lambda}}}{d} \notag \\
\text{s.t. }&M_b \ge \sum_{\lambda} D_{\mbb{P}}(b|\lambda){G}_{\lambda}\quad\quad \forall b \notag \\
&{G}_{\lambda}\ge \mbf{0}\quad \forall \lambda.
\end{align}

The nonclassical fraction 
takes its maximum value, 1, for any nonclassical measurement with rank-one effects (as such effects cannot be nontrivially decomposed in terms of other effects). In fact, {\em all} of the examples we gave in the previous section are rank-one measurements, and so this quantifier has no power to discriminate which of those measurements is more or less nonclassical (other than the fact that it assigns value 0 to the two measurements we considered that were classical).

The dual formulation of this SDP is
\begin{align}
\omega=\max_{\{F_{b}\}_{b}}&~1-\tr\sum_bF_bM_b \notag \\
\text{s.t. }&\sum_{b}D_{\mbb{P}}(b|\lambda)F_b\ge \frac{\mbb{1}}{d}\quad\quad \forall \lambda \notag \\
&F_b\ge \mbf{0} \quad\quad \forall b.
\label{eq: weight-dual}
\end{align}
This can be used to construct an optimal linear witness for a given nonclassical measurement $\{M_b\}_b$, in a manner analogous to the construction of steering witnesses~\cite{Skrzypczyk2014}. In particular, consider any set of Hermitian matrices $\{F_b\}_b$ that is a feasible solution to the SDP. Since $F_b\ge \mbf{0}$, these can be renormalized to generate a set of density operators $\{\rho_b^F:=\frac{1}{f_b}F_b\}_b$.

In the prepare-measure scenario defined by measuring $\{M_b\}_b$ on the states $\{\rho_b^F\}_b$, the violation of inequality $\sum_bf_b\tr[\rho_b^FM_b] \ge 1$ serves as a witness for the nonclassicality of the measurement $\{M_b\}_b$. This is because any {\em classical} measurement $\{K_b\}_b$ (with nonclassical fraction $\omega_{\{K_b\}} =0$) satisfies
\begin{align}
0=\omega_{\{K_b\}}
&:=\max_{\{F'_b\}_b}\left(1-\tr\sum_bF'_bK_b\right)\notag\\
&\geq 1-\tr\left[\sum_b f_b\rho_b^F K_b\right],\notag\\
\Rightarrow\quad
&\sum_b f_b\tr[\rho_b^F K_b]\geq 1.
\end{align}
So if one finds a measurement for which the inequality is violated, then one can conclude that the measurement is nonclassical. Moreover, the witness constructed using the optimal solution to the SDP is an optimal witness.

 One should emphasize, however, that the inequalities obtained from the dual SDP are not noncontextuality inequalities in the usual theory-independent sense~\cite{Schmid2018}. Rather, they are theory-dependent witnesses: violations of such inequalities certify nonclassicality only relative to the quantum description assumed in the derivation, in the same sense that entanglement witnesses are not device-independent. In particular, a violation of such an inequality constitutes a proof of nonclassicality of one's measurement only under the assumption that the states used in the experiment are genuinely the specific quantum states $\{\rho_b^F\}_b$ appearing in the derivation above. In contrast, standard theory-independent noncontextuality inequalities are derived without making any assumptions about the nature of the states in the experiment beyond the operational identity relations that hold among them. Moreover, violations of these witness-based inequalities are not impossible in all noncontextual ontological models. Indeed, in Appendix~\ref{appendix:GPTmodel}, we give an example of a nonclassical measurement, a witness that certifies its nonclassicality if quantum theory were assumed, and a noncontextual ontological model that nevertheless reproduces all the statistics of the prepare-and-measure scenario defined by that witness together with the measurement. 

This is analogous to how entanglement witnesses are device-dependent---they certify that one's state is entangled {\em provided} that one has access to well-characterized quantum measurements. If one wishes to certify nonclassicality of measurements in a device-independent manner, one must instead use a different approach, like the one we discuss in the next section. 

\subsection{Theory-independent certification of nonclassicality of measurements}\label{thryindep}
\label{sec:IVC}
If one wishes to find theory-independent witnesses of nonclassicality, one need look no further than standard noncontextuality inequalities, as studied in, for example, Refs.~\cite{Schmid2018, mazurek2016experimental,schmid2024PTM}. Prior works used violations of noncontextuality inequalities to witness the nonclassicality of an entire scenario rather than of any given single process. However, it is evident that one can use such inequalities to witness the nonclassicality of individual processes as well. The only challenge in doing so is that one cannot assume that every individual process in a given circuit that violates a noncontextuality inequality is itself nonclassical; {\em some} components of the circuit are necessarily nonclassical, but not necessarily {\em all} of them are nonclassical. We elaborate on the question of when one can be certain that a given process is implicated in a proof of nonclassicality in Ref.~\cite[Sec. III]{zhang2024parellel}. For instance, in the simple context of a prepare-measure experiment on a unipartite system, one can conclude from any violation of a theory-independent noncontextuality inequality that the set of measurements in the scenario is nonclassical and that the set of states in the scenario is nonclassical. 

Our work in Ref.~\cite{zhang2024parellel} also leads naturally to an interesting new class of questions regarding what operational means are necessary and/or sufficient to certify the nonclassicality of a given process in a theory-independent way.

Consider the analogous questions in the context of Bell nonclassicality. It is well-known that the entanglement of a given state cannot always be detected via violations of a standard Bell inequality; for example, local measurements on some entangled Werner states do not lead to the violations of any Bell inequalities. So this approach to theory-independent certification of entanglement does not allow one to identify the boundary between entangled (nonfree in the resource theory of Local Operations and Shared Randomness, or LOSR~\cite{sq,Schmid2023understanding,Schmid2020typeindependent}) and separable (LOSR-free). Only by introducing more complicated causal structures can one find theory-independent witnesses of the entanglement of an arbitrary entangled state (see, for example, Ref.~\cite{bowles2018device}, or Section 8 of Ref.~\cite{Schmid2023understanding} for more details).

Here, it is clear that one does not need any causal structure beyond that of a prepare-measure scenario to witness (in a theory-independent way) the nonclassicality present in a given measurement, by the very fact that nonclassicality is defined with respect to the statistics arising in such a scenario (for all possible quantum states). However, it is not clear whether one can always witness the nonclassicality of a particular measurement using a {\em finite} number of states, rather than actually needing to consider all possible quantum states. 

This is particularly unclear for measurements that are near the classical-nonclassical boundary (analogous to how Werner states with a small amount of entanglement do not violate Bell inequalities). However, we now give an example where just five states are sufficient to witness the nonclassicality of a particular measurement under any amount of noise that does not fully destroy its nonclassicality. So at least in some cases, one {\em can} delineate the exact boundary between classical and nonclassical measurements in a theory-independent manner.

\begin{example} 
Consider the $5$-outcome symmetric planar measurement $\msf{M}_5$ defined as in Eq.~\eqref{eq: planar} and studied in Ref.~\cite{Selby2023}. The results of Ref.~\cite{Selby2023} imply that the nonclassicality of this measurement can be certified in a theory-independent manner in a scenario {\em with only five preparations} by violating the noncontextual inequality
\begin{equation}
 q(p_{1|0} + p_{1|2}) + (q-1)p_{2|0} + p_{0|2} - (q+1)p_{1|1}\ge 0,
 \label{ineq:pentagon}
\end{equation}
where $p_{b|a}=\tr[M_b\rho_a]$ and $q=\frac{\sqrt{5}+1}{2}$, if one uses preparations that satisfy the same pentagonal symmetry---that is, that satisfy operational identities of the same form as those satisfied by the effects. 

Consider now the family of noisy measurements where one implements this measurement with probability $\eta$ and implements the trivial measurement with probability $1-\eta$. 
As we showed in the previous section, every measurement in this family is nonclassical if $\eta>\frac{\sqrt{5}-1}{2}$. It turns out that the nonclassicality of {\em every} such measurement can be witnessed by a violation of a noncontextuality inequality using just five states, {\em provided one chooses the states appropriately}. 
The specific five quantum states that were used in Ref.~\cite{Selby2023}
\begin{equation}
 \rho_a=\frac{1}{2}[\mbb{1}+\cos\theta_a\sigma_x+\sin\theta_a\sigma_z],~~~\theta_a=\frac{2\pi a}{5} \label{prep:normal}
\end{equation}
are {\em only} sufficient to verify the nonclassicality of measurements in the family if $\eta>0.764$. If one instead uses a rotated set of quantum states defined as
\begin{equation}
 \rho_a=\frac{1}{2}[\mbb{1}+\cos\theta_a\sigma_x+\sin\theta_a\sigma_z],~~~\theta_a=\frac{2\pi a}{5}+\frac{\pi}{5}\label{prep:tilde},
\end{equation}
and considers the noncontextuality inequality 
\begin{align}
p_{3|0}+(q-1)p_{1|0}-p_{0|0}+qp_{0|1}-p_{1|1}+p_{1
 |3}\ge 0
 \label{ineq:pentagon1}
\end{align}
(obtained from Farkas' lemma using linear programming in Ref.~\cite{Schmid2018}), 
then one can instead violate the inequality for {\em all} nonclassical measurements---for all $\eta>\frac{\sqrt{5}-1}{2}$.
\end{example}

More work is needed to understand whether examples of this sort are generic or not. We have checked numerically that for the measurements discussed in Example~\ref{example1} and Example~\ref{example2}, one can always find a finite set of preparations that is optimal for testing the nonclassicality of that measurement through noncontextual inequalities---the preparations are simply pure quantum states having Bloch vectors corresponding to the dual of the polytope defined by Bloch vectors of the measurements (see footnote~\ref{dual} for the definition of a polytope's dual). On this basis, we make the following conjecture. 

\begin{conjecture}
A finite number of preparations is sufficient for the theory-independent certification of any nonclassical measurement with a finite number of outcomes (through the violation of some noncontextual inequality).
\label{conj:nc-ineq}
\end{conjecture}

\section{Nonclassicality of a set of states}
\label{sec:V}
In the previous section, we focused on defining, certifying, and quantifying nonclassical measurements. We now study the analogous questions, but for the nonclassicality of procedures of the preparation variety. We will do so more succinctly, since the methods are entirely analogous.

 We begin by defining the nonclassicality of a multi-source, which includes a single source and a set of states as special cases, following Ref.~\cite{zhang2024parellel}. 
\begin{definition}
\label{def_multi-sourceNC}
A multi-source is classical if and only if the statistics generated by the set of circuits that contract it with \emph{any} effect (i.e., where the set ranges over all effects) are classically explainable in the sense of Def.~\ref{def:stat}. 
\end{definition}
The analogous conservative interpretation applies to sources. A multi-source is counted as classical only if no quantum effect, and hence no quantum measurement, can make it participate in a prepare-and-measure experiment whose statistics fail to be classically explainable. Consequently, if a multi-source $\msf P$ and a measurement $\msf M$ generate nonclassical statistics,
then $\msf P$ is necessarily nonclassical in the present sense. If $\msf P$ were classical, its contraction with every quantum effect, including the effects of $\msf M$, would be classically explainable. Thus, in a prepare-and-measure scenario, nonclassical statistics certify the nonclassicality of both
components in their respective process-level senses.

Then, following Theorem~2 of Ref.~\cite{zhang2024parellel}, one has the following characterization. 
\begin{theorem}\label{theoremprep}
 A multi-source $\msf{P}\coloneqq\{\{p(a|x)\rho_{a|x}\}_{a}\}_x$ is classical if and only if each of its unnormalized states 
$p(a|x)\rho_{a|x}$ can be decomposed as 
 \begin{align}
 \label{eq_rhodecom} {p(a|x)\rho_{a|x} = \sum_{\lambda} p(a,\lambda|x)\sigma_{\lambda} \quad \forall a,x}
 \end{align}
 for some fixed set of normalized states $\{\sigma_{\lambda}\}_{\lambda}$ and some conditional probability distribution $p(a,\lambda|x)$ 
 satisfying 
 \begin{align}
\sum_{a,x}\alpha_{a,x}p(a,\lambda|x)&=0,
\label{eq_paxOp}\\[-1mm]
&\hspace{-3em}\forall\{\alpha_{a,x}\}_{a,x}\in\mc O(\msf P).
\notag
\end{align}
for all $\lambda$, where $\mc{O}({\msf{P}})$ is defined in Eq.~\eqref{eq:op-prep}. 
\end{theorem}

\begin{proof}
 For completeness, we briefly recall the proof, adapted from the proof of Theorem~2 in Ref.~\cite{zhang2024parellel}.

If the multi-source $\msf{P}$ is classical, then by Def.~\ref{def_multi-sourceNC}, for every quantum effect $M$,
\begin{align}
p(a|x)\tr[M\rho_{a|x}]
&=\sum_{\lambda}p(M|\lambda)p(a,\lambda|x),\notag\\
&\hspace{6em}\forall a,x,M.
\end{align}
where, for each $\lambda$, the assignments $\{p(a,\lambda|x)\}_{a,x}$ satisfy the ontological identities associated with $\mc{O}(\msf{P})$. Since $p(M|\lambda)$ is linear in the effect $M$, positive, and normalized on the unit effect, the generalized Gleason theorem~\cite{Busch2003} implies that there exists a normalized state $\sigma_\lambda$ such that
\begin{align}
p(M|\lambda)=\tr[M\sigma_\lambda]
\quad \forall M.
\end{align}
Hence
\begin{align}
p(a|x)\tr[M\rho_{a|x}]
&=\tr\!\left[M\Bigl(\sum_\lambda
p(a,\lambda|x)\sigma_\lambda\Bigr)\right],\notag\\
&\hspace{6em}\forall a,x,M.
\end{align}
and therefore
\begin{align}
p(a|x)\rho_{a|x}=\sum_\lambda p(a,\lambda|x)\sigma_\lambda
\qquad \forall a,x.
\end{align}
Moreover, the ontological identities on $p(a,\lambda|x)$ are exactly those in Eq.~\eqref{eq_paxOp}.

Conversely, if Eq.~\eqref{eq_rhodecom} holds with $p(a,\lambda|x)$ satisfying Eq.~\eqref{eq_paxOp}, define
\begin{align}
p(M|\lambda)\coloneqq \tr[M\sigma_\lambda].
\end{align}
Then $p(M|\lambda)$ is a valid response function, linear in $M$, and
\begin{align}
p(a|x)\tr[M\rho_{a|x}]
=\sum_\lambda p(M|\lambda)p(a,\lambda|x).
\end{align}
Thus, the statistics generated by contracting $\msf{P}$ with arbitrary quantum effects are classically explainable, so $\msf{P}$ is classical.
\end{proof}

 Unlike the measurement case, a direct polytope characterization is not available for general multi-sources. We therefore proceed in two steps. First, using Proposition~\ref{prop:to_single_modelstates}, we reduce the qualitative problem to the case of a set of normalized states. Second, in that reduced setting, Theorem~\ref{theoremprep} naturally leads to a preparation-assignment polytope, allowing us to develop the quantitative theory. The obstruction for general multi-sources was already noted in Ref.~\cite[Sec.~III]{schmid2024PTM}. 

\begin{proposition}
A multi-source $\{\{p(a|x)\rho_{a|x}\}_{a}\}_{x}$\footnote{ Throughout this work, we assume $p(a|x)>0$ for all $a,x$; otherwise, one should restrict to the support of $p(a|x)$. } is nonclassical if and only if the set of states $\{\rho_{a|x}\}_{a,x}$ is nonclassical.
\label{prop:to_single_modelstates}
\end{proposition}

 Indeed, as noted in Refs.~\cite{selby2023accessible,zhang2024parellel}, in a prepare-and-measure scenario, the statistics generated by a multi-source $\msf{P}\coloneqq\{\{p(a|x)\rho_{a|x}\}_{a}\}_x$ together with an arbitrary measurement $\msf{M}$ are classically explainable if and only if the statistics generated by the flag-convexified set of states $\{\rho_{a|x}\}_{a,x}$ together with $\msf{M}$ are classically explainable. Hence, by Def.~\ref{def_multi-sourceNC}, qualitative nonclassicality is preserved under flag-convexification. 

 Accordingly, in the remainder of this section, we focus on sets of states, which we denote simply by $\{\rho_a\}_a$; the extension back to general multi-sources is discussed in Appendices~\ref{appfc1} and~\ref{appfc2}. Let $k$ denote the number of states in the set. By the same flag-convexification logic as above, $\{\rho_a\}_a$ is classical if and only if the uniform source $\{\frac{1}{k}\rho_a\}_a$ is classical. Applying Theorem~\ref{theoremprep} to this uniform source, Eq.~\eqref{eq_rhodecom} becomes
\begin{align}
\frac{1}{k}\rho_a
= \sum_\lambda p(a,\lambda)\sigma_\lambda
= \sum_\lambda p(a|\lambda)p(\lambda)\sigma_\lambda
\qquad \forall a.
\label{eqkdec}
\end{align}
Equivalently,
\begin{align}
\rho_a=\sum_\lambda p(a|\lambda)\tilde{\sigma}_\lambda
\qquad \forall a,
\end{align}
where $\tilde{\sigma}_\lambda\coloneqq k\,p(\lambda)\sigma_\lambda$. Thus, once we restrict attention to sets of states, the relevant ontological data are the conditional probabilities $\{p(a|\lambda)\}_a$, exactly as in the measurement case. This motivates the following definition. 

\begin{definition}
\label{def: op_polytope_prep}
A noncontextual preparation-assignment polytope $\mbb{P}_{\msf{P}}$ associated with a set of states $\{\rho_{a}\}_{a}$ is the set of vectors $\{p(a)\}_{a}$ that satisfy the constraints 
\begin{subequations}
\begin{align}
\text{(i)}~~&p(a) \geq 0~~~~\forall a, \label{eq:cons_posp}\\
\text{(ii)}~~&\sum_a p(a)=1, \label{eq:cons_normp}\\
\text{(iii)}~~&\sum_{a}{\alpha}_{a} p(a)=0\quad\quad \forall \{{\alpha}_{a}\}\in \mc{O}({\msf{P}}), \label{eq:cons_opp}
\end{align}
\end{subequations}
where, for a set of states $\{\rho_a\}_a$, $\mc{O}(\msf{P})$ denotes the operational identities among the states, i.e.,
\begin{align}
\mc{O}(\msf{P})\coloneqq \Bigl\{\{\alpha_a\}_a \,\Big|\, \sum_a \alpha_a \rho_a =0\Bigr\}.
\end{align}
\end{definition}

Following logic exactly like that in Eq.~\eqref{eqdetm} to write the decomposition in terms of the extremal points of the polytope just defined, we obtain the following:
\begin{corollary}
\label{coro: nc-stat}
A set of states $\{\rho_a\}_{a}$ is classical if and only if it can be decomposed as
 \begin{align}
 \label{eqCoro2}
 \rho_{a}&=\sum_{\lambda} D_{\mbb{P}}(a|\lambda)\tilde{\sigma}_{\lambda} \quad \quad \forall a,
 \end{align}
 where $\{\tilde{\sigma}_{\lambda}\}_{\lambda}$ is a set of unnormalized states,
 and $D_{\mbb{P}}(a|\lambda)$ are the extremal points in the noncontextual preparation-assignment polytope $\mbb{P}_{\msf{P}}$ for the set of states. 
\end{corollary}

\section{Certifying and quantifying the nonclassicality of a set of states}
\label{sec:VI}
In this section, we show how one can certify and quantify the nonclassicality of a set of states $\{\rho_a\}_{a}$. With only small notational modifications (given in Appendix~\ref{appfc2}), our approach can be applied to general multi-sources rather than sets of states; however, we have opted to focus on the latter special case due to the fact that they are much more commonly studied. 

We again introduce a quantifier based on the robustness of nonclassicality to white noise, and a weight-based quantifier. The former of these can be used directly to quantify the nonclassicality of general multi-sources, since (as we prove in Appendix~\ref{appfc1}) the white-noise robustness is unchanged under moving from a given multi-source to its associated set of normalized states. Whether the latter is preserved under this move is not known, and so minor modifications would need to be made to apply our results to general multi-sources (as we discuss in Appendix~\ref{appfc2}).

Like for measurements, the nonclassicality of a set of states $\{\rho_a\}_{a}$ can be certified using an SDP written in terms of the finitely many extreme points $D_{\mbb{P}}(a|\lambda)$ 
of the polytope in Def.~\ref{def: op_polytope_prep}; namely, the SDP
\begin{align}
\max_{\{\tilde{\sigma}_{\lambda}\}_{\lambda}}&~ \mu \notag \\
\text{s.t. }&\sum_{\lambda}D_{\mbb{P}_{\msf{P}}}(a|\lambda)\tilde{\sigma}_{\lambda}=\rho_a\quad\quad \forall a \notag \\
&\tilde{\sigma}_{\lambda}\ge \mu\mbb{1} \quad \forall \lambda.
\end{align}
If the program returns an optimal value $\mu$ that is negative, then $\{\rho_a\}_{a}$ is nonclassical. Otherwise, it is classical.

\subsection{Robustness-based nonclassicality quantifier for a set of states}

One can also quantify nonclassicality of sources using a quantifier based on robustness to noise, just as we did above for measurements. Again, we take the white-noise robustness as an example, and we consider mixing each state with white noise, as
\begin{equation}
\rho_{a}^{\eta} = \eta \rho_{a} + (1-\eta)\frac{1}{d}\mathbb{1},
\end{equation}
where $d$ is the Hilbert space dimension. 
When $\eta = 0$, every element of the ensemble is proportional to the maximally mixed state, and so the source is always classical. The critical parameter $\eta$ at which the transition to classicality occurs is the white-noise robustness of the source, and can be computed using the following semidefinite program:
\begin{align}
\eta_{\{\rho_a\}}= \max_{\{\tilde{\sigma}_{\lambda}\}_{\lambda}} &\eta \notag \\
\text{s.t. }&\sum_{\lambda}D_{\mbb{P}}(a|\lambda)\tilde{\sigma}_{\lambda}=\rho^{\eta}_a\quad\quad \forall a \notag \\
& \eta\le 1,~\tilde{\sigma}_{\lambda}\ge 0\quad \forall \lambda.
\end{align}
Again, we can get an analytical upper bound on this measure by studying the dual of the primal problem above, which is 
\begin{align}
\eta_{\{\rho_a\}}=& \min_{\{X_{a}\}_a} 1 + \sum_{a}\tr[X_{a} \rho_{a}] \notag \\
&\text{s.t. }1 + \sum_{a} \tr[X_{a} \rho_{a}]\geq \frac{1}{d}\sum_{a} \tr X_{a}, \notag \\
&\sum_{a} D_{\mbb{P}}(a|\lambda) X_{a} \geq 0 \quad \forall \lambda.
\end{align}

By optimizing a cleverly chosen family of dual variables of the form $X_a=\alpha\mbb{1}-\beta \rho_a$ following the same procedure as in Ref.~\cite{Designolle2019} (and Appendix~\ref{appendix: upperbound}),
one can obtain an analytical upper bound on the white-noise robustness, namely 
\begin{align}
 \eta_{\{\rho_a\}}\le \frac{kd\Lambda-k}{d\sum_{a=1}^k {\tr \rho_a^2}- k},
 \label{upper bound}
\end{align}
where $k$ is the number of states in the set $\{\rho_a\}_{a}$ and $\Lambda=\max_{\lambda}\norm{\sum_a D_{\mbb{P}}(a|\lambda)\rho_a}_{2}$. For a set of $k$ pure quantum states $\{\rho_a\}_{a=1}^{k}$, this further simplifies to
\begin{align}
\eta\le \frac{d\Lambda - 1}{d-1}.
\end{align}

\begin{example}
The set $\{\ket{0}, \ket{1}, \ket{+}, \ket{-}\}$ of four states used in the standard BB84 protocol has white-noise robustness $\eta=\frac{1}{\sqrt{2}}\approx 0.7071$.
\end{example}

\begin{example}
The set $\{\ket{0}, \ket{1}, \ket{+}, \ket{-}, \ket{+i}, \ket{-i}\}$ of states in the 6-state QKD protocol~\cite{dagmar1998} has white-noise robustness $\eta=\frac{1}{\sqrt{3}}\approx 0.5774$.
\end{example}

\begin{example}
 The set $\{\ket{0}, \ket{1}, \frac{1}{2}\ket{0}\pm \frac{\sqrt{3}}{2}\ket{1}, \frac{\sqrt{3}}{2}\ket{0}\pm \frac{1}{2}\ket{1}\}$ studied in Ref.~\cite{gencontext} has white-noise robustness $\eta=\frac{1}{\sqrt{3}}\approx 0.5774$. 
\end{example}
\begin{example}
The set of states $\{\rho_x=\frac{1}{2}(\mbb{1}+ \hat{n}_x\cdot\vec{\sigma})\}_{x=1}^{8}$, where $\{ \hat{n}_x\}_{x=1}^{8}$ are unit vectors corresponding to the eight
vertices of a regular cube inscribed in the Bloch sphere has white-noise robustness $\eta=\frac{1}{\sqrt{3}}\approx 0.5774$. 
\label{exam: cube}
\end{example}
\begin{example}
The set of states $\{\rho_x=\frac{1}{2}(\mbb{1}+ \hat{n}_x\cdot\vec{\sigma})\}_{x=1}^{12}$, where $\{ \hat{n}_x\}_{x=1}^{12}$ are unit vectors corresponding to the twelve vertices of a regular icosahedron inscribed in the Bloch sphere, has white-noise robustness $\eta=\sqrt{\frac{1+q^2}{3q^4}}\approx 0.4195$, where $q=\frac{\sqrt{5}+1}{2}$. 
\label{exam: icosahedron}
\end{example}
\subsection{Nonclassical fraction for preparations}

Like for measurements, one can study the minimal cost required to generate the set of states $\{\rho_a\}_a$ in question by mixing together an arbitrary nonclassical set of states with an arbitrary classical set of states such that the sets of states being mixed satisfy the same operational identities as the given set of states. (The constraint on operational identities is included to enforce convexity, for the same reasons as in Section~\ref{ncfraction}. ) That is, the states in any given set $\{\rho_a\}_a$ can be written as
\begin{equation}
 \rho_a=\omega \gamma_a + (1-\omega) \kappa_a, \label{eq:frac_decomp} 
\end{equation}
where $0\le \omega\le 1$, where $\{\gamma_a\}_a$ is an arbitrary set of states while $\{\kappa_a\}_a$ is classical, and where they both satisfy the same operational identities as $\{\rho_a\}_a$. The minimal value of $\omega$ that can be used in such a decomposition is the nonclassical fraction of the set of states. This quantity can be computed using the following SDP:
\begin{align}
\omega=\min_{\{\tilde{\sigma}_{\lambda}\}_{\lambda}} &1-\frac{\tr\sum_{\lambda}\tilde{\sigma}_{\lambda}}{k} \notag \\
\text{s.t. }&\rho_a \ge \sum_{\lambda} D_{\mbb{P}}(a|\lambda)\tilde{\sigma}_{\lambda}\quad\quad \forall a \notag \\
& \tilde{\sigma}_{\lambda}\ge 0\quad \forall \lambda.
\end{align}
Note that the nonclassical fraction 
takes its maximum value, 1, for any nonclassical set of pure states. 

The dual formulation of this SDP is
\begin{align}
\omega=\max_{\{F_{a}\}_{a}}&~1-\frac{\tr\sum_aF_a\rho_a}{k} \notag \\
\text{s.t. }&\sum_{a}D_{\mbb{P}}(a|\lambda)F_a\ge \mbb{1}\quad\quad \forall \lambda \notag \\
&F_a\ge 0 \quad\quad \forall a.
\label{eq: weight-dual-state}
\end{align}
This formulation can be used to construct an optimal linear witness for a given nonclassical set of states $\{\rho_a\}_a$. Given a feasible solution $\{F_a\}_a$, one can rescale $F_a=f_a\tilde{F}_a$ using any $f_a$ such that $\tilde{F}_a\le \mbb{1}$, and then introduce the set of two-outcome measurements $\{\tilde{F}_a, \mbb{1}-\tilde{F}_a\}_a$. This set of measurements constitutes a witness for the nonclassicality of the given set of states. One considers the prepare-and-measure experiment with the $k$ states $\{\rho_a\}_{a=1}^k$ and the $k$ two-outcome measurements $\{\{\tilde{F}_a, \mbb{1}-\tilde{F}_a\}\}_{a=1}^k$; if the inequality 
\begin{equation}
\sum_a f_a\tr[\tilde{F}_a\rho_a]\ge 1 
\end{equation}
is violated, then the set of states is nonclassical. If the feasible solution is optimal, then the witness is optimal.

\section{Using nonclassicality of assemblages to witness entanglement}
\label{sec:VII}

Having developed witnesses for nonclassical sets of states, we now apply them to assemblages, where they yield entanglement certification results beyond steering-based certification.

A special case of a multi-source is a steering assemblage~\cite{Wiseman2007}. Steering assemblages are distinguished from general multi-sources by the additional constraint $\sum_a p(a|x)\rho_{a|x}=\sigma$, where $\sigma$ is a fixed state that does not depend on the setting variable $x$. This extra constraint is often called the no-signaling principle. 

We now reprove a result of Ref.~\cite{Plavala2024}, stated in passing right after their Theorem 2 (and where we have rephrased the result in the language of nonclassical sets of states). 

\begin{corollary}
 A two-qubit state is entangled if and only if it can be steered to a nonclassical set of states.
\end{corollary}

\begin{proof}
It was shown in Ref.~\cite{Jevtic2014, Nguyen2016} that for any separable two-qubit state, the set of states one can steer to on either side can be embedded in a tetrahedron inside the Bloch ball; this set of states is classical, as they fit inside a simplex inside the quantum state space (see Corollary 1 in Ref.~\cite{zhang2024parellel}). So, 
separable states can only be steered to classical sets of states.

The other direction is given by Theorem 1 in Ref.~\cite{Plavala2024}, which states that a bipartite state is separable if there is a noncontextual model for the prepare-measure scenario involving the set of all states to which the other subsystem can be steered, together with all quantum effects---that is, if the set of all states to which the other subsystem can be steered is classical.
\end{proof}
This corollary should be contrasted with the well-known fact that entanglement is necessary but {\em not} sufficient for being steered to an assemblage that is nonfree in the resource theory of LOSR nonclassicality~\cite{Schmid2020typeindependent,Zjawin2023quantifyingepr,Zjawin2023resourcetheoryof}. The difference arises because
nonclassicality of an assemblage in the sense of LOSR
(namely, steerability) is not implied by nonclassicality in the sense defined here. 

It follows that even for two-qubit states whose entanglement cannot be certified by witnessing the steerability of the assemblages it can generate, one can always certify its entanglement by studying the nonclassicality (in the sense defined in this work and in Ref.~\cite{zhang2024parellel}) of the assemblages it can generate. We now give an explicit example.

\begin{example}
\label{exam: ineq}
 Consider the family of noisy isotropic states given by $\rho_{\text{Iso}}^{\eta}=\eta\op{\Psi^+}{\Psi^+}+(1-\eta)\frac{1}{4}\mbb{1}$. It has recently been proven that for $\eta\le 1/2$, the state is local and unsteerable~\cite{zhang2024}. However, if one performs a measurement with effects forming the vertices of a regular icosahedron, namely $N_{\pm|x}=\frac{1}{2}(\mbb{1}\pm \hat{n}_x\cdot\vec{\sigma})$ where $\{\pm \hat{n}_x\}_{x=1}^6$, then the corresponding normalized states in the resulting assemblage are the noisy version of that in Example~\ref{exam: icosahedron}. Hence, by Prop.~\ref{prop:to_single_modelstates}, the resulting assemblage is nonclassical for $\eta>\sqrt{\frac{5-2\sqrt{5}}{3}}\approx 0.4195$. One can also certify the nonclassicality of this bipartite state in a theory-independent manner using (for example) the inequality (and quantum violation) that we give in Appendix~\ref{examineqapp}.

\end{example}
Beyond the two-qubit case, however, one cannot in general guarantee that a state is entangled simply by verifying that it can be steered to a nonclassical set of states (as already noted in Ref.~\cite{Plavala2024} and articulated clearly in Ref.~\cite{zhang2025all}). Further investigation of these ideas is left for future work.

Theorem 2 of Ref.~\cite{Tavakoli2020} states that `an assemblage is unsteerable if and only if its statistics admits a preparation and measurement noncontextual model for all measurements'. However, our work shows that this claim is not accurate. For example, the assemblage $\{p(a|x)\rho_{a|x}:=\tr_A[(N_{a|x}\otimes \mbb{1}) \rho_{\text{Iso}}^{\eta}] \}_{a,x}$ in Example~\ref{exam: ineq} for $0.4195\approx\sqrt{\frac{5-2\sqrt{5}}{3}}<\eta<1/2$~\cite{zhang2024,Renner2024} is unsteerable but nonclassical, and so there is no preparation and measurement noncontextual model for its statistics. The mistake arises because the purported proof of this theorem does not take into account all possible operational identities in the scenario it considers, but rather only those derived from the no-signaling principle. 
But other operational identities typically arise in such scenarios; we give an exhaustive list in Eqs.~42(a)-42(e) of Ref.~\cite{zhang2024parellel}.

\section{Discussion}

We introduced two measures of nonclassicality: one that quantifies how robust a given process is to white noise, and a weight-based quantifier. A natural next step would be to see if an entire resource theory can be developed to systematically and exhaustively quantify the nonclassicality of any given process.

Robustness-based quantifiers are often linked to discrimination tasks~\cite{Designolle2019}. Consequently, another natural question is whether the measures we introduced here are closely linked to any interesting information-theoretic tasks for which nonclassicality of a given multi-measurement or multi-source offers a quantum advantage.

A final interesting question for future work is to consider `liftings' of noncontextual inequalities, analogous to the lifting of Bell inequalities~\cite{Pironio2005}.

\section*{Author contributions}
All authors contributed to the conceptual development, technical analysis, and preparation of the manuscript. Large language models were used for language editing, organizational suggestions, critical feedback on the exposition, and assistance with reference checking. All scientific claims, derivations, numerical results, and the final text were independently verified and approved by the authors.

\section*{Acknowledgements}
We acknowledge valuable discussions with Robert Spekkens. YZ thanks Elie Wolfe and Ravi Kunjwal for many useful discussions. YZ is supported by the Natural Sciences and Engineering Research Council of Canada (NSERC) and the Canada Foundation for Innovation (CFI). YZ is grateful for the hospitality of Perimeter Institute, where this work was carried out. DS and YY were supported by Perimeter Institute for Theoretical Physics. Research at Perimeter Institute is supported in part by the Government of Canada through the Department of Innovation, Science and Economic Development and by the Province of Ontario through the Ministry of Colleges and Universities. YY was also supported by the Natural Sciences and Engineering Research Council of Canada (Grant No. RGPIN-2024-04419).

\appendix
\section{Flag-convexification preserves white-noise robustness}\label{appfc1}

Proposition~\ref{prop:to_single_model} states that the qualitative nonclassicality of a set of measurements coincides with that of its flag-convexification. We now show that flag-convexification also does not change the white-noise robustness of sets of quantum measurements. (We expect that general measures of nonclassicality will not remain invariant under flag-convexification, but it remains an open question.) 
\begin{proposition}
The white-noise robustness $\eta_{\{M_{b|y}\}}$ for a set of measurements $\{\{M_{b|y}\}_{b}\}_{y}$ coincides with the white-noise robustness $\eta_{\{\tilde{M}_{b,y}\}}$ of its flag-convexification $\{\tilde{M}_{b,y}\coloneqq\frac{1}{|Y|}M_{b|y}\}_{b,y}$.
\label{prop:to_single_modelquantitative}
\end{proposition}
\begin{proof}
As a consequence of Proposition~\ref{prop:to_single_model}, any $\eta$-noisy multi-measurement $\{\{
\eta M_{b|y}+(1-\eta)\frac{\tr[M_{b|y}]}{d}\mbb{1} \}_{b}\}_{y}$ is nonclassical if and only if $
\{\eta\tilde{M}_{b,y}+(1-\eta)\frac{\tr[\tilde{M}_{b,y}]}{d}\mbb{1}\}_{b,y}$ is nonclassical, which is just the flag-convexified version of the $\eta$-noisy multi-measurement. Since this equivalence holds for every $\eta$, the white-noise robustness of any set of measurements remains unchanged under flag-convexification.
\end{proof}

Similarly, by using Proposition~\ref{prop:to_single_modelstates}, we can show that flag-convexification also does not change the white-noise robustness of sets of quantum preparations.

\begin{proposition}
\label{prop:robustate}
The white-noise robustness $\eta_{\{p(a|x)\rho_{a|x}\}}$ for a multi-source $\{p(a|x)\rho_{a|x}\}$ coincides with the white-noise robustness $\eta_{\{\rho_{a|x}\}}$ of the corresponding set of normalized states $\{\rho_{a|x}\}$.
\end{proposition}
\begin{proof}
This follows directly from Proposition~\ref{prop:to_single_modelstates}: the multi-source $\{\{p(a|x)\rho^{\eta}_{a|x}\}_a\}_x$ is nonclassical if and only if the set of states $\{\rho^{\eta}_{a|x}\}_{a,x}$ is nonclassical for any $\eta$, where $\rho^{\eta}_{a|x}=\eta\rho_{a|x}+(1-\eta)\frac{\mbb{1}}{d}$. Since this equivalence holds for every $\eta$, the claim follows.
\end{proof}

Similarly, the white-noise robustness of a multi-source is equivalent to that of the source obtained from it by flag-convexification. Indeed, one can by a similar proof show that the white-noise robustness of both multi-measurements and multi-sources is unchanged under flag-convexification by an {\em arbitrary} (non-uniform) full-support flag-convexification~\cite{selby2023accessible}. 

\section{Flag-convexification preserves the nonclassical fraction} \label{appfc2}
\begin{proposition}
\label{flag-conv}
The nonclassical fraction $\omega_{\{M_{b|y}\}}$ for a set of measurements $\{\{M_{b|y}\}_{b}\}_{y}$ coincides with the nonclassical fraction $\omega_{\{\tilde{M}_{b,y}\}}$ of its flag-convexification $
\{\tilde{M}_{b,y}
\coloneqq 
\frac{1}{|Y|}M_{b|y}\}_{b,y}$.
\end{proposition}
\begin{proof}
This follows from the fact that if there exists a set of measurements $\{\{N_{b|y}\}_b\}_y$ and a classical set of measurements $\{\{K_{b|y}\}_b\}_y$ (both respecting the same operational identities as $\{\{M_{b|y}\}_{b}\}_{y}$) such that 
$$M_{b|y}=\omega_{\{M_{b|y}\}} N_{b|y}+(1-\omega_{\{M_{b|y}\}})K_{b|y}$$
One can directly define the corresponding flag-convexified measurements $\{\tilde{N}_{b,y}\}_{b,y}\coloneqq\{\frac{1}{|Y|}N_{b|y}\}_{b,y}$ and $\{\tilde{K}_{b,y}\}_{b,y}\coloneqq\{\frac{1}{|Y|}K_{b|y}\}_{b,y}$, which have the same classicality/nonclassicality, respectively, according to Proposition~\ref{prop:to_single_model}; these evidently satisfy the same operational identities as $\{\tilde{M}_{b,y}\}$, as well as
$$\tilde{M}_{b,y}=\omega_{\{M_{b|y}\}} \tilde{N}_{b,y}+(1-\omega_{\{M_{b|y}\}})\tilde{K}_{b,y}$$
Thus, $\omega_{\{\tilde{M}_{b,y}\}}\le \omega_{\{M_{b|y}\}}$. (This is just an upper bound rather than an equality, since we have constructed one specific decomposition of $\tilde{M}_{b,y}$ from the given decomposition of ${M}_{b|y}$, but have not necessarily found the optimal decomposition).\par 
For the other direction, since any $\{\tilde{N}_{b,y}\}$ and $\{\tilde{K}_{b,y}\}$ must necessarily satisfy $\sum_b \tilde{N}_{b,y}=\sum_b\tilde{K}_{b,y}=\frac{1}{|Y|}\mbb{1}$ as they respect the same operational identity as $\{\tilde{M}_{b,y}\}$. By similar logic, one could always construct $\{\{{N}_{b|y}=|Y|\tilde{N}_{b,y}\}_b\}_y$ and $\{\{{K}_{b|y}=|Y|\tilde{K}_{b,y}\}_b\}_y$ as decomposition of $\{\{M_{b|y}\}_b\}_y$. One can thus show that $\omega_{\{\tilde{M}_{b,y}\}}\ge \omega_{\{M_{b|y}\}}$, and therefore we have $\omega_{\{\tilde{M}_{b,y}\}}= \omega_{\{M_{b|y}\}}$.
\end{proof}

One can similarly prove the analogous result for flag-convexification of sources.

The nonclassical fraction of a source $\{\{p(a|x)\rho_{a|x}\}_a\}_x$ can be defined analogously to the definition we gave for sets of states; namely, it is the smallest number $0\le \omega\le 1$, such that one can write
\begin{equation}
p(a|x)\rho_{a|x}=\omega p'(a|x)\gamma_{a|x} + (1-\omega) p''(a|x)\kappa_{a|x}, \label{eq:frac_decomp1} 
\end{equation}
where $\{\{p'(a|x)\gamma_{a|x} \}_a\}_x$ is an arbitrary multi-source, $\{\{p''(a|x)\kappa_{a|x}\}_a\}_x$ is a classical multi-source, and both of these satisfy the same operational identities as $\{\{p(a|x)\rho_{a|x}\}_a\}_x$. 

\begin{proposition}
The nonclassical fraction $\omega_{\{p(a|x)\rho_{a|x}\}}$ for a multi-source $\{\{p(a|x)\rho_{a|x}\}_a\}_{x}$ coincides with the nonclassical fraction $\omega_{\{p(a,x)\tilde{\rho}_{a|x}\}}$ of its flag-convexification $\{p(a,x)\tilde{\rho}_{a|x}=\frac{1}{|X|}p(a|x)\rho_{a|x}\}_{a,x}$.
\end{proposition}

The proof is identical to that of Prop.~\ref{flag-conv}, replacing the measurement label $y$ by the source-setting label $x$.

We note that $p(a|x)$, $p'(a|x)$, and $p''(a|x)$ need not be the same. Therefore, it is not obvious that the nonclassical fraction defined for a general multi-source coincides with the nonclassical fraction for the set of renormalized states associated with it. In other words, flag-convexification preserves the nonclassical fraction of a multi-source, but this does not imply equality with the nonclassical fraction of the associated normalized-state set, because renormalization changes operational identities.

\section{Nonclassicality of measurements in mutually unbiased bases}
\label{appendix:MUB}

In a complex Hilbert space of dimension $d$, a set of projective measurements
$\{\{M_{b|y}\}_{b=1}^d\}_y$ is called \textit{mutually unbiased} if
\begin{equation}
 \tr[M_{b|y} M_{b'|y'}] = \frac{1}{d}
\end{equation}
for all $b,b'$ whenever $y\neq y'$.
The robustness of incompatibility of mutually unbiased bases under noise has been studied extensively~\cite{Bavaresco2017,Designolle2019}.

We now show that, for measurements obtained from mutually unbiased bases by adding white noise,
nonclassicality coincides with incompatibility. More precisely, consider the noisy effects
\begin{equation}
\widetilde M_{b|y}^{(\eta)} \coloneqq \eta M_{b|y} + (1-\eta)\frac{\mbb{1}}{d},
\qquad 0\le \eta \le 1.
\label{eq:MUB-noisy}
\end{equation}
We will show that for $0<\eta\le 1$, the noncontextual measurement-assignment polytope
$\mbb{P}_{\msf{M}}$ associated with the noisy MUB measurements is simply $\mbb{P}_{\msf{M}}=\bigoplus_{y}\Delta_d$, where $\Delta_d$ is the $d$-simplex defined only by positivity and normalization.
In that case, Eq.~\eqref{eq_Edecom2} imposes no restrictions beyond stochasticity, and Eq.~\eqref{eq_Edecom1} reduces exactly to the standard condition for joint measurability. Hence, for $0<\eta\le 1$, these noisy MUB measurements are nonclassical if and only if they are incompatible.

The key point is that the only linear dependence relations among the effects of mutually unbiased bases are those implied by POVM normalization, namely
\begin{equation}
\sum_{b=1}^d M_{b|y}=\sum_{b'=1}^d M_{b'|y'}
\qquad \forall y,y'.
\label{eq:MUB-normalization}
\end{equation}
These relations imply no constraints on the response functions beyond normalization $\sum_b p(b|y,\lambda)=1$, $\forall y,\lambda$.

We first establish the corresponding statement for the sharp MUB measurements $\{M_{b|y}\}$, and then transfer it to the noisy ones $\{\widetilde M_{b|y}^{(\eta)}\}$.

\begin{lemma}
Let $\{\{M_{b|y}\}_{b=1}^d\}_{y=1}^n$ be a collection of mutually unbiased bases in dimension $d$. Then the set
\begin{equation}
\{M_{b|1}\}_{b=1}^d \cup \bigcup_{y=2}^n \{M_{b|y}\}_{b=1}^{d-1}
\label{eq:MUB-key-set}
\end{equation}
is linearly independent.
\end{lemma}

\begin{proof}
Suppose, for contradiction, that there exist coefficients $\alpha_{b,y}$, not all zero, such that
\begin{equation}
0=\sum_{b=1}^d \alpha_{b,1} M_{b|1}
+\sum_{y=2}^n \sum_{b=1}^{d-1} \alpha_{b,y} M_{b|y}.
\label{eq:MUB-key-dependence}
\end{equation}
Taking the Hilbert--Schmidt inner product of Eq.~\eqref{eq:MUB-key-dependence} with any included effect $M_{b'|y'}$ gives
\begin{align}
\alpha_{b',y'}+\sum_{y\neq y'}\frac{1}{d}\sum_b \alpha_{b,y}=0,
\label{eq:MUB-key-inner}
\end{align}
where the inner summation runs over $b=1,\dots,d$ for $y=1$ and over $b=1,\dots,d-1$ for $y\ge 2$.

Since the second term is independent of $b'$, it follows that for each fixed $y'$, the coefficient $\alpha_{b',y'}$ is independent of $b'$. Hence there exist scalars $\alpha_{y'}$ such that
\begin{equation}
\alpha_{b',y'}=\alpha_{y'}
\end{equation}
for all included effects in basis $y'$.

Equation~\eqref{eq:MUB-key-dependence} therefore reduces to a linear system
\begin{equation}
A\alpha=0,
\end{equation}
where
\begin{equation}
A=
\begin{bmatrix}
1 & \frac{d-1}{d} & \frac{d-1}{d} & \cdots & \frac{d-1}{d} \\
1 & 1 & \frac{d-1}{d} & \cdots & \frac{d-1}{d} \\
1 & \frac{d-1}{d} & 1 & \cdots & \frac{d-1}{d} \\
\vdots & \vdots & \vdots & \ddots & \vdots \\
1 & \frac{d-1}{d} & \frac{d-1}{d} & \cdots & 1
\end{bmatrix},
\end{equation}
Subtracting the first row from each of the remaining rows yields an upper-triangular matrix with diagonal entries $[1,\frac{1}{d},\dots,\frac{1}{d}]$, all of which are nonzero. Hence $A$ is invertible, so $\alpha=0$. This contradicts the assumption of a nontrivial linear dependence. Therefore the set in Eq.~\eqref{eq:MUB-key-set} is linearly independent.
\end{proof}

We can now determine all linear dependence relations among the full family $\{\{M_{b|y}\}_{b=1}^d\}_y$. If there are $n$ bases, then the full family contains $nd$ effects, while the linearly independent set in Eq.~\eqref{eq:MUB-key-set} contains
\begin{equation}
d+(n-1)(d-1)=nd-(n-1)
\end{equation}
effects. Therefore the space of linear relations among the full family has dimension at most $n-1$.

On the other hand, the $n-1$ normalization relations
\begin{equation}
\sum_{b=1}^d M_{b|y}-\sum_{b=1}^d M_{b|1}=0,
\qquad y=2,\dots,n,
\label{eq:MUB-normalization-independent}
\end{equation}
are clearly linearly independent since each relation is the unique one involving the effect $M_{d|y}$ for a given $y\ge 2$. Hence these $n-1$ relations exhaust all linear dependence relations among the effects. Therefore the only operational identities among the MUB effects are those implied by normalization, and so $\mbb{P}_{\msf{M}}=\bigoplus_y\Delta_d$. It follows from Theorem~\ref{theoremmeas} that, for sharp mutually unbiased bases, nonclassicality is equivalent to incompatibility.

Finally, consider the noisy MUB effects in Eq.~\eqref{eq:MUB-noisy} with $0<\eta\le 1$. Since
\begin{equation}
\sum_b \widetilde M_{b|y}^{(\eta)}=\mbb{1},
\end{equation}
the identity operator lies in the span of the noisy effects. Therefore, for every $b,y$,
\begin{equation}
M_{b|y}
=
\frac{1}{\eta}\widetilde M_{b|y}^{(\eta)}
-\frac{1-\eta}{\eta d}\mbb{1},
\end{equation}
so each sharp effect lies in the span of the noisy ones. The converse inclusion is immediate from Eq.~\eqref{eq:MUB-noisy}, and hence the noisy and sharp families have the same linear span. Consequently they satisfy exactly the same operational identities, and therefore the associated noncontextual measurement-assignment polytope is again
$\mbb{P}_{\widetilde{\msf{M}}}=\bigoplus_y \Delta_d$.
Thus, for noisy mutually unbiased bases with $0<\eta\le 1$, nonclassicality is equivalent to incompatibility as well.

More generally, the same conclusion holds whenever the linear spans of the effects of two measurements intersect only in $\operatorname{span}\{\mbb{1}\}$. In that case, there are no operational identities relating the two measurements beyond those implied by normalization, and the notion of nonclassicality studied here again reduces to incompatibility.

\section{Upper bounding the white-noise robustness of nonclassical measurements}
\label{appendix: upperbound}
The dual SDP for the robustness test we discussed in the main text is
\begin{align}
\eta_{\{M_b\}}
&=\min_{\{X_b\}_b}\left(1+\sum_b\tr[X_bM_b]\right)\notag\\
\text{s.t.}\quad
&1+\sum_b\tr[X_bM_b]\notag\\
&\qquad\geq\frac{1}{d}\sum_b\tr X_b\,\tr M_b,\notag\\
&\sum_bD_{\mbb P}(b|\lambda)X_b\geq0
\quad\forall\lambda.
\end{align}
By weak duality, the value of any feasible dual solution upper-bounds the primal optimum; moreover, all feasible solutions $\{X_b\}_b$ to the above dual problem provide an upper bound on the white-noise robustness $\eta_{\{M_b\}}$. Consider a family of dual variables $X_b=\alpha\mbb{1}-\beta M_b$~\cite{Designolle2019}. Such an $X_b$ is feasible if the two constraints in the dual SDP are satisfied, so
\begin{align}
&1-\beta\sum_b [\tr M_b^2- \frac{1}{d}(\tr M_b)^2]\ge 0\notag\\
&\alpha\mbb{1}-\beta \sum_b D_{\mbb{P}}(b|\lambda)M_b\ge 0, \forall \lambda.
\end{align}
Define $\Lambda\coloneqq\max_{\lambda}\norm{\sum_b D_{\mbb{P}}(b|\lambda)M_b}_{2}$, where $\norm{X}_{2}$ is the spectral norm defined as the largest absolute value of the eigenvalues of $X$. When the above two inequalities are saturated, we obtain
\begin{subequations}
\begin{align}
 &\beta=\frac{1}{\sum_b [\tr M_b^2- \frac{1}{d}(\tr M_b)^2]}\\
&\alpha=\max_{\lambda}\norm{\sum_b D_{\mbb{P}}(b|\lambda)M_b}_{2}\beta=\beta\Lambda, 
\end{align}
\end{subequations} 
so the optimal dual variable in this class is
\begin{align}
 X_b=\frac{\Lambda\mbb{1}-M_b}{\sum_b [\tr M_b^2- \frac{1}{d}(\tr M_b)^2]}.
\end{align}
Plugging this into the objective function, we can obtain a nontrivial upper bound on the robustness:
\begin{align}
 \eta\le \frac{d^2\Lambda-\sum_b (\tr{M_b})^2}{\sum_b [d{\tr M_b^2}- ({\tr M_b})^2]}
\end{align}
For a rank-1 POVM $\{M_b\}_b$ in a $d$-dimensional space with $k$ elements having equal trace $\tr [M_b]=\frac{d}{k}$, this further simplifies to
\begin{align}
\eta\le \frac{k\Lambda - 1}{d-1}.
\end{align}

\section{Explicit demonstration that nonclassical measurement witnesses are theory-dependent}\label{appendix:GPTmodel}

As noted in the main text, the nonclassicality witnesses we introduced in Section~\ref{ncfraction} are theory-dependent. That is, a violation of those inequalities {\em only} constitutes a proof of nonclassicality of one's measurement {\em assuming that the states in one's experiment are genuinely the specific quantum states $\{\rho_b^F\}_b$ assumed in the above derivation.} Moreover, it is not the case that violations of these inequalities are impossible to generate in any noncontextual ontological model. We now demonstrate this by giving an example of a nonclassical measurement, a nonclassicality witness that certifies its nonclassicality, and a noncontextual ontological model that reproduces the statistics of the prepare-measure scenario defined by the witness together with the measurement.

Take the BB84 measurement $\msf{M}_4=\{M_b\}_{b=1}^4$ as an example: 
\begin{equation}
 M_b=\frac{1}{4}[\mbb{1}+\cos\theta_b\sigma_x+\sin\theta_b\sigma_z]~~~\theta_b=\frac{\pi b}{2}.
\end{equation}
The corresponding optimal dual variables $\{F_b\}_b$, obtained by solving the dual SDP in Eq.~\eqref{eq: weight-dual}, form a set of unnormalized quantum states with Bloch vectors anti-parallel to the corresponding measurements $\{M_b\}_b$, i.e., 
\begin{align}
F_b=\frac{2+\sqrt{2}}{4}[\mbb{1}-\cos\theta_b\sigma_x-\sin\theta_b\sigma_z]~~~~\theta_b=\frac{\pi b}{2}.
\end{align}
By definition, any classical measurement $K_b$ (that has the same operational identity as $\msf{M}_4$) must obey
\begin{align*}
 \sum_b \tr[F_bK_b] \ge 1.
\end{align*}
However, for the noisy BB84 measurement $\msf{M}^{\eta}_4=\{M^{\eta}_b\}_{b=1}^4$ with $\eta>\frac{1}{\sqrt{2}}$, (which is nonclassical as mentioned in Example~\ref{bb84}), we have
\begin{equation} 
 \sum_b \tr[F_bM^{\eta}_b]<1.
\end{equation}

By defining a set of states $\{\rho_b\}_b$ with $\rho_b:=\frac{2}{2+\sqrt{2}}F_b$, we obtain a nonclassicality witness of the form discussed in Section~\ref{ncfraction}. In particular, for any classical set of measurement $\{K_b\}_{b=1}^4 $, the statistics one can generate in a prepare-and-measure experiment whose states are taken to be $\{\rho_b\}_b$ must satisfy
\begin{align}
 \sum_b p(b|\rho_b):=\sum_b \tr[\rho_bK_b]\ge 2-\sqrt{2}.
 \label{eq: linear witness}
\end{align}
However, the BB84 measurement can achieve 
\begin{equation}\label{violation1}
 \sum_b p(b|\rho_b)=0,
\end{equation} 
and so is nonclassical. The violation of this inequality does not imply the impossibility of a noncontextual model for such a prepare-measure experiment, however; indeed, we will now give a noncontextual model for the BB84 measurement together with this set of states. The quantum statistics in such a prepare-measure scenario are
\begin{align}
P(b|\rho_1)&=\{0,\frac{1}{4},\frac{1}{2},\frac{1}{4}\},~P(b|\rho_2)=\{\frac{1}{4},0,\frac{1}{4}, \frac{1}{2}\} \notag \\ 
P(b|\rho_3)&=\{\frac{1}{2}, \frac{1}{4},0,\frac{1}{4}\},~P(b|\rho_4)=\{\frac{1}{4},\frac{1}{2},\frac{1}{4}, 0\}, 
\end{align}
where the four numbers in each set correspond to the probabilities for $b=1,2,3,4$, respectively.
Plugging in the relevant probabilities, one can verify that quantum theory achieves Eq.~\eqref{violation1} and so violates the inequality in Eq.~\eqref{eq: linear witness}.

One can reproduce the quantum predictions for this scenario in the following noncontextual model (which is simply the Spekkens toy theory in Ref.~\cite{Spekken2007}).
The epistemic states are defined as
\begin{align}
&\mu(\lambda_i|\rho_1)=\{0,\frac{1}{2},\frac{1}{2},0\},~\mu(\lambda_i|\rho_2)=\{\frac{1}{2},0, \frac{1}{2},0\}\notag \\
&\mu(\lambda_i|\rho_3)=\{\frac{1}{2},0,0, \frac{1}{2}\},~\mu(\lambda_i|\rho_4)=\{0, \frac{1}{2},0, \frac{1}{2}\},
\end{align}
and the response functions for the measurement are defined as
\begin{align}
&p(b|\lambda_1)=\{\frac{1}{2},0 ,0,\frac{1}{2}\},~p(b|\lambda_2)=\{0,\frac{1}{2},\frac{1}{2},0\}\notag \\
&p(b|\lambda_3)=\{0,0,\frac{1}{2},\frac{1}{2}\},~p(b|\lambda_4)=\{\frac{1}{2},\frac{1}{2},0,0\}.
\end{align}

One can check directly that this model reproduces the quantum predictions and also the operational identities in the scenario. 

This shows explicitly that the witness in question is theory-dependent, as there exist noncontextual theories that can violate the bound in question. Only if one assumes the correctness of quantum theory (and in particular, has an exact characterization of one's states) does a violation of this bound certify the nonclassicality of the measurement in question.

\section{Noncontextuality inequality for Example~\ref{exam: ineq}}
\label{examineqapp}

Consider the prepare-and-measure scenario obtained by steering the noisy isotropic state from Example~\ref{exam: ineq}. The corresponding statistics are
\begin{align}
p(ab|xy)&=p(a|x)\tr[M_{b|y}\rho_{a|x}^{\eta}], \label{eq:appF-pmstats}\\
p(a|x)\rho_{a|x}^{\eta}&=\tr_A\!\big[(N_{a|x}\otimes\mbb{1})\rho_{\mathrm{Iso}}^{\eta}\big],
\label{eq:appF-assemblage}
\end{align}
where
\begin{align}
N_{\pm|x}=\frac{1}{2}(\mbb{1}\pm \hat{n}_x\cdot\vec{\sigma}),
~~
M_{\pm|y}=\frac{1}{2}(\mbb{1}\pm \hat{m}_y\cdot\vec{\sigma})^T.
\end{align}
Here, $\{\pm \hat{n}_x\}_{x=0}^5$ and $\{\pm \hat{m}_y\}_{y=0}^{9}$ correspond, respectively, to the vertices of a regular icosahedron and its dual dodecahedron. 
It can be verified directly, for instance using the linear program of Ref.~\cite{Selby2024}, that the resulting behavior $p(ab|xy)$ is nonclassical whenever
\begin{equation}
\eta>\sqrt{\frac{1+q^2}{3q^4}}\approx 0.4195,
\end{equation}
where $q=\frac{\sqrt{5}+1}{2}$.

Using Farkas' lemma together with the tools of Ref.~\cite{Schmid2018}, we derived the following inequality, which is satisfied by every noncontextual behavior obeying the same operational identities as those induced by $\{p(a|x)\rho_{a|x}^{\eta}\}$ and $\{M_{b|y}\}$:
\begin{align}
&p(00|02)+p(00|12)+p(00|22)-p(00|01)\notag\\
&\qquad -p(00|10)-p(00|23)\le \frac{1}{q^2}\approx 0.3820.
\label{eq:ineq-app}
\end{align}

To exhibit a quantum violation, take
\begin{align}
\hat n_0&=\frac{1}{\sqrt{1+q^2}}[-1,-q,0],\notag\\
\hat n_1&=\frac{1}{\sqrt{1+q^2}}[-q,0,1],\notag\\
\hat n_2&=\frac{1}{\sqrt{1+q^2}}[-q,0,-1],\notag\\
\hat m_0&=\frac{1}{\sqrt3}[-1,-1,-1],\notag\\
\hat m_1&=\frac{1}{\sqrt3}[-q,q^{-1},0],\notag\\
\hat m_2&=\frac{1}{\sqrt3}[-q,-q^{-1},0],\notag\\
\hat m_3&=\frac{1}{\sqrt3}[-1,-1,1].
\end{align}
For these choices, the left-hand side of Eq.~\eqref{eq:ineq-app} evaluates to $\frac{3\eta}{\sqrt{3(1+q^2)}}$. Hence Eq.~\eqref{eq:ineq-app} is violated precisely when
\begin{align}
\frac{3\eta}{\sqrt{3(1+q^2)}}&>\frac{1}{q^2}
\notag\\
&\Longleftrightarrow\quad
\eta>\sqrt{\frac{1+q^2}{3q^4}}\approx0.4195.
\end{align}
At $\eta=1$, the corresponding quantum value is
\begin{equation}
\frac{3}{\sqrt{3(1+q^2)}}\approx 0.9106.
\end{equation}

\bibliographystyle{quantum}
\bibliography{ref}
\end{document}